\newif{\ifarxiv}
\arxivtrue%

\ifarxiv%
\documentclass[a4paper]{article}
\else%
\documentclass[a4paper,UKenglish,numberwithinsect]{lipics-v2016}
\fi%

\newcommand{\ifnotarxiv}{\ifarxiv\else}

\ifarxiv%
\usepackage[T1]{fontenc}
\usepackage{lmodern}

\usepackage{authblk}

\usepackage{amsmath}
\usepackage{amssymb}
\usepackage{amsthm}
\usepackage{aliascnt}
\usepackage{hyperref}
\usepackage[margin=25mm]{geometry}

\newtheorem{theorem}{Theorem}[section]

\newaliascnt{lemma}{theorem}
\newtheorem{lemma}[lemma]{Lemma}
\aliascntresetthe{lemma}

\newaliascnt{corollary}{theorem}
\newtheorem{corollary}[corollary]{Corollary}
\aliascntresetthe{corollary}

\theoremstyle{definition}

\newaliascnt{definition}{theorem}
\newtheorem{definition}[definition]{Definition}
\aliascntresetthe{definition}
\fi%

\usepackage{microtype}
\usepackage{mathpartir}
\usepackage{paralist}
\usepackage{subcaption}

\usepackage{thmtools}
\usepackage{thm-restate}

\newcommand{\naturals}{\mathbb{N}}

\newcommand{\pom}{\ensuremath{\mathsf{Pom}_\Sigma}}
\newcommand{\pomsp}{\ensuremath{\pom^\mathsf{sp}}}
\newcommand{\width}[1]{\|#1\|}
\newcommand{\depth}[1]{\left(\!\left|#1\right|\!\right)}
\newcommand{\terms}{\ensuremath{{\mathcal{T}_\Sigma}}}
\newcommand{\qterms}{\ensuremath{{\mathcal{Q}_\Sigma}}}

\DeclareFontFamily{U}{matha}{}
\DeclareFontShape{U}{matha}{m}{n}{%
  <-5.5>    matha5
  <5.5-6.5> matha6
  <6.5-7.5> matha7
  <7.5-8.5> matha8
  <8.5-9.5> matha9
  <9.5-11>  matha10
  <11->     matha12
}{}
\DeclareSymbolFont{matha}{U}{matha}{m}{n}
\DeclareFontSubstitution{U}{matha}{m}{n}
\DeclareFontFamily{U}{mathx}{\hyphenchar\font45}
\DeclareFontShape{U}{mathx}{m}{n}{<-> mathx10}{}
\DeclareSymbolFont{mathx}{U}{mathx}{m}{n}
\DeclareFontSubstitution{U}{mathx}{m}{n}

\DeclareMathDelimiter{\ldbrack}{4}{matha}{"76}{mathx}{"30}
\DeclareMathDelimiter{\rdbrack}{5}{matha}{"77}{mathx}{"38}

\newcommand{\sem}[1]{\left\ldbrack#1\right\rdbrack}

\newcommand{\project}[2]{{#1}_{#2}}

\newcommand{\atrace}[2][]{\mathrel{\raisebox{-3pt}{$\xrightarrow[#1]{#2}$}}}
\newcommand{\artrace}[2][]{\mathrel{\raisebox{-3pt}{$\xrightarrow[#1]{#2}\mkern-14mu\rightarrow$}}}
\newcommand{\tracerel}{\mathrel{\rightarrow}}
\newcommand{\rtracerel}{\mathrel{\rightarrow\mkern-14mu}\rightarrow}

\newcommand{\lmbrace}{\{\mskip-4mu|}
\newcommand{\rmbrace}{|\mskip-4mu\}}
\newcommand{\mset}[1]{\lmbrace#1\rmbrace}
\newcommand{\angl}[1]{\left\langle#1\right\rangle}
\newcommand{\restr}[2]{\ensuremath{#1\!\upharpoonright_{#2}}}

\newcommand{\satrace}[1]{\atrace{#1}_\Sigma}
\newcommand{\sartrace}[1]{\artrace{#1}_\Sigma}
\newcommand{\satracerel}{\tracerel_\Sigma}
\newcommand{\sartracerel}{\rtracerel_\Sigma}
\newcommand{\ssderiv}{\delta_\Sigma}
\newcommand{\psderiv}{\gamma_\Sigma}
\newcommand{\sacc}{F_\Sigma}
\newcommand{\slang}{L_\Sigma}

\newcommand{\qsatrace}[1]{\atrace{#1}_\simeq}
\newcommand{\qsatracerel}{\tracerel_\simeq}
\newcommand{\qssderiv}{\delta_\simeq}
\newcommand{\qpsderiv}{\gamma_\simeq}
\newcommand{\qsacc}{F_\simeq}
\newcommand{\qssupport}{\pi_\simeq}
\newcommand{\qslang}{L_\simeq}

\newcommand{\pipe}{\;|\;}

\usepackage{tikz}
\usetikzlibrary{positioning}
\usetikzlibrary{automata}

\ifarxiv%
\pgfmathsetmacro{\scale}{0.7}
\else%
\pgfmathsetmacro{\scale}{0.6}
\fi%

\title{Brzozowski Goes Concurrent {\Large A~Kleene~Theorem~for~Pomset~Languages}}

\ifarxiv%
\author[1]{Tobias Kapp\'e\thanks{\texttt{tkappe@cs.ucl.ac.uk}}}
\date{}
\else%
\author[1]{Tobias Kapp\'e}
\fi%

\author[1]{Paul Brunet}
\author[2]{Bas Luttik}
\author[1]{Alexandra Silva}
\author[1]{Fabio Zanasi}

\ifarxiv%
\affil[1]{University College London, London, United Kingdom}
\else%
\affil[1]{University College London, London, United Kingdom \\ \texttt{tkappe@cs.ucl.ac.uk}}
\fi%
\affil[2]{Eindhoven University of Technology, Eindhoven, The Netherlands}

\ifnotarxiv%
\titlerunning{Brzozowski Goes Concurrent --- A~Kleene~Theorem~for~Pomset~Languages}

\authorrunning{T. Kapp\'e, P. Brunet, B. Luttik, A. Silva and F. Zanasi}
\Copyright{Tobias Kapp\'e, Paul Brunet, Bas Luttik, Alexandra Silva and Fabio Zanasi}
\fi%

\ifnotarxiv%

\EventEditors{Roland Meyer and Uwe Nestmann}
\EventNoEds{2}
\EventLongTitle{28th International Conference on Concurrency Theory (CONCUR 2017)}
\EventShortTitle{CONCUR 2017}
\EventAcronym{CONCUR}
\EventYear{2017}
\EventDate{September 5--8, 2017}
\EventLocation{Berlin, Germany}
\EventLogo{}
\SeriesVolume{85}
\ArticleNo{21}
\fi%

\ifarxiv%
\newenvironment{arxivproof}{\proof}{\endproof}
\else%
\usepackage{comment}
\excludecomment{arxivproof}
\fi%

\begin{document}

\maketitle

\begin{abstract}
Concurrent Kleene Algebra (CKA) is a mathematical formalism to study programs that exhibit concurrent behaviour.
As with previous extensions of Kleene Algebra, characterizing the free model is crucial in order to develop the foundations of the theory and potential applications.
For CKA, this has been an open question for a few years and this paper makes an important step towards an answer.
We present a new automaton model and a Kleene-like theorem that relates a relaxed version of CKA to series-parallel pomset languages, which are a natural candidate for the free model.
There are two substantial differences with previous work: from expressions to automata, we use Brzozowski derivatives, which enable a direct construction of the automaton; from automata to expressions, we provide a syntactic characterization of the automata that denote valid CKA behaviours.
\end{abstract}

\section{Introduction}

In their CONCUR'09 paper~\cite{hoare-moeller-struth-wehrman-2009}, Hoare, M\"{o}ller, Struth, and Wehrman introduced Concurrent Kleene Algebra (CKA) as a suitable mathematical framework to study concurrent programs, in the hope of achieving the same elegance that Kozen did when using Kleene Algebra (and extensions) to provide a verification platform for sequential programs.

CKA is a seemingly simple extension of Kleene Algebra (KA): it adds a parallel operator that allows to specify concurrent behaviours compositionally.
However, extending the existing KA toolkit --- importantly, completeness and decidability results --- turns out to be challenging.
A fundamental missing ingredient is a characterization of the free model for CKA\@.
This is in striking contrast with KA, where these topics are well understood.
Several authors~\cite{hoare-staden-moeller-struth-zhu-2016,jipsen-moshier-2016} have conjectured the free model to be series-parallel pomset languages --- a generalization of regular languages to sets of partially ordered words.

In KA, Kleene's theorem provided a pillar for developing the toolkit and axiomatization~\cite{kozen-1997}, and, by extension, characterizing the free model.
In this light, we pursue a Kleene Theorem for CKA\@.
Specifically, we study series-rational expressions, with a denotational model in terms of pomset languages.
Our main contribution is a Kleene Theorem for series-rational expressions, based on constructions faithfully translating between the denotational model and a newly defined operational model, which we call \emph{pomset automata}.
In a nutshell, these are finite-state automata in which computations from a certain state $s$ may branch into parallel threads that contribute to the language of $s$ whenever they both reach a final state.

\medskip
We are not the first to attempt such a Kleene theorem.
However, earlier works~\cite{lodaya-weil-2000,jipsen-moshier-2016} fall short of giving a precise correspondence between the denotational and operational models, due to the lack of a suitable automata restriction ensuring that only valid behaviours are accepted.
We overcome this situation by introducing a generalization of Brzozowski derivatives~\cite{brzozowski-1964} in the translation from expressions to automata.
This guides us to a \emph{syntactic} restriction on automata (rather than the \emph{semantic} condition put forward in previous works), which guarantees the existence of a reverse construction, from automata to expressions.
Moreover, following the Brzozowski route allows us to bypass a Thompson-like construction~\cite{thompson-1968}, avoiding the introduction of $\epsilon$-transitions and non-determinism present in the aforementioned works.

Since series-parallel expressions do not include the parallel analogue of the Kleene star (the ``parallel star''), and our denotational model is not sound for the exchange law (which governs the interaction between sequential and parallel composition), our contribution is most accurately described as an operational model for \emph{weak Bi-Kleene Algebra}.
We leave it to future work to extend our construction to work with a denotational model that is sound for the exchange law (thus moving to \emph{weak Concurrent Kleene Algebra}), as well as add the parallel star operator (arriving at Concurrent Kleene Algebra proper).

The remainder of this paper is organized as follows.
In Section~\ref{section:preliminaries}, we introduce the necessary notation.
In Section~\ref{section:pomset-automata}, we introduce our automaton model as well as some notable subclasses of automata.
In Section~\ref{section:expressions-to-automata}, we discuss how to translate a series-rational expression to a semantically equivalent pomset automaton, while in Section~\ref{section:automata-to-expressions} we show how to translate a suitably restricted class of pomset automata to series-rational expressions.
We contrast results with earlier work in Section~\ref{section:discussion}.
Directions for further work in are listed in Section~\ref{section:further-work}.

\ifarxiv%
To preserve the flow of the narrative, proofs of the more routine lemmas appear in Appendix~\ref{appendix:proofs}.
\else%
To save space, proofs of lemmas are omitted from this paper.
For a discussion that includes a proof for every lemma, we refer to the extended version of this paper~\cite{kappe-brunet-luttik-silva-zanasi-2017-arxiv}.
\fi%

\subparagraph*{Acknowledgements}
We thank the anonymous reviewers for their insightful comments.
This work was partially supported by the ERC Starting Grant ProFoundNet (grant code 679127).

\section{Preliminaries}%
\label{section:preliminaries}

Let $S$ be a set; we write $2^S$ for the set of all subsets of $S$, and $\binom{S}{2}$ for the set of \emph{multisets} over $S$ of size two.
An element of $\binom{S}{2}$ containing $s_1, s_2 \in S$ is written $\mset{s_1, s_2}$; note that $\mset{s_1, s_2} = \mset{s_2, s_1}$, and that $s_1$ may be the same as $s_2$.
We use the symbols $\phi$ and $\psi$ to denote multisets.
If $S$ and $I$ are sets, and for every $i \in I$ there exists an $s_i \in S$, we call ${(s_i)}_{i \in I}$ an \emph{$I$-indexed family over $S$}.
We say that a relation $\prec\ \subseteq S \times S$ is a \emph{strict order} on $S$ if it is irreflexive and transitive.
We refer to $\prec$ as \emph{well-founded} if there are no infinite descending $\prec$-chains, i.e., no family ${(s_n)}_{n\in\naturals}$ over $S$ such that $\forall n \in \naturals, s_{n+1}\prec s_n$.
Throughout the paper we fix a finite set $\Sigma$ called the \emph{alphabet}, whose elements are symbols usually denoted with $a$ and $b$.
Lastly, if $\rightarrow\ \subseteq X \times Y \times Z$ is a ternary relation, we write $x \xrightarrow{y} z$ instead of $\angl{x, y, z} \in\ \rightarrow$.

\subsection{Pomsets}

\emph{Partially-ordered multisets}, or \emph{pomsets}~\cite{gischer-1988} for short, generalise words to a setting where events (elements from $\Sigma$) may take place not just sequentially, but also in parallel.

\begin{definition}%
\label{definition:pomset}
A \emph{labelled poset} is a tuple $\angl{U, \leq_U, \lambda_U}$ consisting of a \emph{carrier} set $U$, a partial order $\leq_U$ on $U$ and a \emph{labelling} function $\lambda_U: U \to \Sigma$.
A \emph{labelled poset isomorphism} is a bijection between poset carriers that bijectively preserves the labels and the ordering.
A \emph{pomset} is an isomorphism class of labelled posets; equivalently, it is a labelled poset up-to bijective renaming of elements in $U$.
We write $1$ for the empty pomset, $\pom$ for the set of all pomsets and $\pom^+$ for the set of all the non-empty pomsets.
\end{definition}

For instance, suppose a recipe for caramel-glazed cookies tells us to
\begin{inparaenum}[(i)]
    \item\label{recipe:dough} \emph{prepare} cookie dough
    \item\label{recipe:bake} \emph{bake} cookies in the oven
    \item\label{recipe:caramelize} \emph{caramelize} sugar
    \item\label{recipe:glaze} \emph{glaze} the finished cookies.
\end{inparaenum}
Here, step~\eqref{recipe:dough} precedes steps~\eqref{recipe:bake} and~\eqref{recipe:caramelize}.
Furthermore, step~\eqref{recipe:glaze} succeeds both steps~\eqref{recipe:bake} and~\eqref{recipe:caramelize}.
A pomset representing this process could be $\angl{C, \leq_C, \lambda_C}$, where $C = \{\eqref{recipe:dough},\eqref{recipe:bake},\eqref{recipe:caramelize},\eqref{recipe:glaze}\}$ and $\leq_C$ is such that $\eqref{recipe:dough}\leq_C\eqref{recipe:bake}\leq_C\eqref{recipe:glaze}$ and $\eqref{recipe:dough}\leq_C\eqref{recipe:caramelize}\leq_C\eqref{recipe:glaze}$; $\lambda_C$ is as in the recipe.

\begin{figure}
    \centering
    \begin{subfigure}[b]{0.2\textwidth}
        \centering
        \begin{tikzpicture}[scale=\scale,every node/.style={transform shape}]
            \node (dough) {\textsf{prepare}};
            \node[above right=10mm of dough,anchor=center] (bake) {\textsf{bake}};
            \node[below right=10mm of dough,anchor=center] (caramelize) {\textsf{caramelize}};
            \node[right=20mm of dough,anchor=center] (glaze) {\textsf{glaze}};
            \draw[->] (dough) edge (bake);
            \draw[->] (dough) edge (caramelize);
            \draw[->] (bake) edge (glaze);
            \draw[->] (caramelize) edge (glaze);
        \end{tikzpicture}
        \caption{Diagram for $C$.}\label{figure:hasse-diagram-cookies}
    \end{subfigure}
    \begin{subfigure}[b]{0.37\textwidth}
        \centering
        \begin{tikzpicture}[scale=\scale,every node/.style={transform shape}]
            \node (dough) {\textsf{prepare}};
            \node[above right=10mm of dough,anchor=center] (bake) {\textsf{bake}};
            \node[below right=10mm of dough,anchor=center] (caramelize) {\textsf{caramelize}};
            \node[right=20mm of dough,anchor=center] (glaze) {\textsf{glaze}};
            \draw[->] (dough) edge (bake);
            \draw[->] (dough) edge (caramelize);
            \draw[->] (bake) edge (glaze);
            \draw[->] (caramelize) edge (glaze);
            \node[right=10mm of glaze,anchor=center] (new-dough) {\textsf{prepare}};
            \node[above right=10mm of new-dough,anchor=center] (new-bake) {\textsf{bake}};
            \node[below right=10mm of new-dough,anchor=center] (new-caramelize) {\textsf{caramelize}};
            \node[right=20mm of new-dough,anchor=center] (new-glaze) {\textsf{glaze}};
            \draw[->] (new-dough) edge (new-bake);
            \draw[->] (new-dough) edge (new-caramelize);
            \draw[->] (new-bake) edge (new-glaze);
            \draw[->] (new-caramelize) edge (new-glaze);
            \draw[->] (glaze) edge (new-dough);
        \end{tikzpicture}
        \caption{Diagram for $C \cdot C$.}\label{figure:hasse-diagram-cookies-repeated}
    \end{subfigure}
    \begin{subfigure}[b]{0.37\textwidth}
        \begin{tikzpicture}[scale=\scale,every node/.style={transform shape}]
            \node[state] (q0) {$q_0$};
            \node[state,right=15mm of q0] (q1) {$q_1$};
            \node[state,above right=7mm of q1] (q3) {$q_3$};
            \node[state,below right=7mm of q1] (q4) {$q_4$};
            \node[state,accepting,right=20mm of q3] (q5) {$q_5$};
            \node[state,accepting,right=20mm of q4] (q6) {$q_6$};
            \node[state,right=15mm of q1] (q2) {$q_2$};
            \node[state,accepting,right=15mm of q2] (q7) {$q_7$};

            \draw[->] (q0) edge node[above] {\textsf{prepare}} (q1);
            \draw[-] (q1) edge (q3);
            \draw[-] (q1) edge (q4);
            \draw[dashed,->] (q1) edge (q2);
            \draw[->] (q3) edge node[above] {\textsf{bake}} (q5);
            \draw[->] (q4) edge node[below] {\textsf{caramelize}} (q6);
            \draw (q1) + (-45:7mm) arc (-45:45:7mm);
            \draw[->] (q2) edge node[above] {\textsf{glaze}} (q7);
        \end{tikzpicture}
        \caption{PA accepting $C$.}\label{figure:example-pa}
    \end{subfigure}
    \caption{Hasse diagrams for pomsets and a pomset automaton accepting one.}
\end{figure}
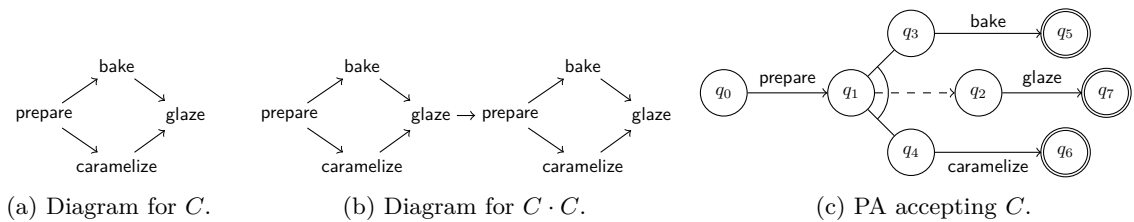
Note that words are just finite pomsets with a total order.
We will sometimes use $a \in \Sigma$ to refer to the pomset with a single point labelled $a$ (and the obvious order); such a pomset is called \emph{primitive}.
A pomset can be represented as a Hasse diagram, where nodes have labels in $\Sigma$.
For instance, the Hasse diagram for the pomset $C$ above is drawn in Figure~\ref{figure:hasse-diagram-cookies}.

To simplify notation, we refer to a pomset by the carrier $U$ of a labelled poset $\angl{U, \leq_U, \lambda_U}$ in its isomorphism class.
We use the symbols $U$, $V$, $W$ and $X$ to denote pomsets.
Pomsets being isomorphism classes, the content of the carrier of the chosen representative is of very little importance; it is the order and labelling that matters.
For this reason, we tacitly assume that whenever we have two pomsets, we pick representatives that have disjoint carrier sets.

\begin{definition}%
\label{definition:pomset-width}
The \emph{width} of a pomset $U$, denoted $\width{U}$, is the size of the largest antichain in $U$ with respect to $\leq_U$, i.e., the maximum $n \in \naturals$ such that there exist $u_1, u_2, \dots, u_n \in U$ that are not related by $\leq_U$.
\end{definition}

The pomsets we work with in this paper have a finite carrier.
As a result, $\width{U}$ is always defined.
For instance, the width of the pomset $C$ above is $2$, because the nodes $\eqref{recipe:bake}$ and $\eqref{recipe:caramelize}$ are an antichain of size $2$, and there is no antichain of size $3$.

\begin{definition}%
\label{definition:pomset-composition}
Let $U$ and $V$ be pomsets.
The \emph{sequential composition} of $U$ and $V$, denoted $U \cdot V$, is the pomset $\angl{U \cup V, \leq_U \cup \leq_V \cup\, (U \times V), \lambda_U \cup \lambda_V}$.
The \emph{parallel composition} of $U$ and $V$, denoted $U \parallel V$, is the pomset $\angl{U \cup V, \leq_U \cup \leq_V, \lambda_U \cup \lambda_V}$.
Here, $\lambda_U \cup \lambda_V$ is the function from $U \cup V$ to $\Sigma$ that agrees with $\lambda_U$ on $U$, and with $\lambda_V$ on $V$.
\end{definition}
Note that $1$ is the unit for both sequential and parallel composition.
Sequential composition forces the events in the left pomset to be ordered before those in the right pomset.
An example, describing the pomset $C \cdot C$, is depicted in Figure~\ref{figure:hasse-diagram-cookies-repeated}.

\begin{definition}%
\label{definition:pomset-series-parallel}
The set of \emph{series-parallel} pomsets, $\pomsp$, is the smallest set that includes the empty and primitive pomsets and is closed under sequential and parallel composition.
\end{definition}

In this paper we will be mostly concerned with series-parallel pomsets.
For inductive reasoning about them, it is useful to record the following lemma.
\begin{lemma}%
\label{lemma:pomset-unique-factorization}
Let $U \in \pomsp$.
If $U$ is non-empty, then \emph{exactly one} of the following is true:
\begin{inparaenum}[(i)]
    \item\label{case:decompose-primitive} $U = a$ for some $a \in \Sigma$, or
    \item\label{case:decompose-sequential} $U = V \cdot W$ for non-empty $V, W \in \pomsp$, strictly smaller than $U$, or
    \item\label{case:decompose-parallel} $U = V \parallel W$ for non-empty $V, W \in \pomsp$, strictly smaller than $U$.
\end{inparaenum}
\end{lemma}
\begin{arxivproof}
We first show that at least one case must hold.
The proof proceeds by induction on the construction of $\pomsp$.
In the base, $U$ is a primitive pomset (since it is non-empty), and thus $U = a$ for some $a \in \Sigma$;~\eqref{case:decompose-primitive} holds.
In the inductive step, there are two cases to consider.
\begin{itemize}
    \item If $U = V \cdot W$ for $V$ and $W$ series-parallel, suppose that either $V$ or $W$ is empty.
    Then $U = W$ or $U = V$ respectively, and thus the claim follows by induction.
    If neither $V$ nor $W$ is empty, they must both be strictly smaller than $U$, and thus~\eqref{case:decompose-sequential} holds.
    \item If $U = V \parallel W$ for $V$ and $W$ series-parallel, suppose that either $V$ or $W$ is empty.
    Then $U = W$ or $U = V$ respectively, and thus the claim follows by induction.
    If neither $V$ nor $W$ is empty, they must both be strictly smaller than $U$, and thus~\eqref{case:decompose-parallel} holds.
\end{itemize}

To see that at most one of~\eqref{case:decompose-primitive},~\eqref{case:decompose-sequential} and~\eqref{case:decompose-parallel} can hold, assume that at least two of them hold.
There are three combinations to consider; we derive a contradiction for each.
\begin{itemize}
    \item If~\eqref{case:decompose-primitive} and~\eqref{case:decompose-sequential} hold, then $a = U = V \cdot W$ for some $a \in \Sigma$ and non-empty $V, W \in \pomsp$, strictly smaller than $U$.
    But then $V$ and $W$ must both be empty (since $U = a$, and thus $U$ must be of size one), and thus it follows that $U$ must be empty as well --- a contradiction.
    \item If~\eqref{case:decompose-primitive} and~\eqref{case:decompose-parallel} hold, then a contradiction is reached by an argument similar to the above.
    \item If~\eqref{case:decompose-sequential} and~\eqref{case:decompose-parallel} hold, then $U = V_1 \cdot W_1$ and $U = V_2 \parallel W_2$ for non-empty $V_1, W_1, V_2, W_2 \in \pomsp$, all strictly smaller than $U$.
    Suppose $v$ is in $V_1$, and $w$ is in $W_1$ (such a $v$ and $w$ exist, since $V_1$ and $W_1$ are non-empty).
    Then $v$ is ordered before $w$ in $V_1 \parallel W_2$, since $V_1 \cdot W_1 = V_2 \parallel W_2$.
    But then $v$ and $w$ are either both in $V_2$ or both in $W_2$ --- if they were not, this would contradict the definition of parallel composition.
    Suppose $v$ and $w$ are both in $V_2$, and let $w'$ be in $W_2$ (such a $w'$ exists, since $W_2$ is non-empty).
    Then $w'$ is ordered neither before nor after $v$ or $w$ in $V_2 \parallel W_2$, and thus the same holds in $V_1 \cdot W_2$.
    But then $w'$ is not in $V_1$ (otherwise, $w'$ would be ordered before $w$), nor in $W_1$ (otherwise, $w'$ would be ordered after $v$).
    This contradicts that $V_1 \cdot W_1 = V_2 \parallel W_2$.
    A similar argument can be made for the case where $v$ and $w$ are both in $W_2$. \qedhere
\end{itemize}
\end{arxivproof}

\subsection{Pomset languages}

If a sequential program can exhibit multiple traces, we can group the words that represent these traces into a set called a \emph{language}.
By analogy, we can group the pomsets that represent the traces that arise from a parallel program into a set, which we refer to as a \emph{pomset language}.
Pomset languages are denoted by the symbols $\mathcal{U}$ and $\mathcal{V}$.

For instance, suppose that the recipe for glazed cookies may has an optional fifth step where chocolate sprinkles are spread over the cookies.
In that case, there are \emph{two} pomsets that describe a trace arising from the recipe, $C^+$ and $C^-$, either with or without the chocolate sprinkles.
The pomset language $\mathcal{C} = \{ C^-, C^+ \}$ describes the new recipe.

\begin{definition}%
\label{defintition:pomset-language-width}
Let $\mathcal{U}$ be a pomset language.
$\mathcal{U}$ has \emph{bounded width} if there is $n \in \naturals$ such that for all $U \in \mathcal{U}$ we have $\width{U} \leq n$.
The minimal such $n$ is the \emph{width} of $\mathcal{U}$, written $\width{\mathcal{U}}$.
\end{definition}
The pomset languages considered in this paper have bounded width, and hence $\width{\mathcal{U}}$ is always defined.
For instance, the width of $\mathcal{C}$ is $2$, because the width of both $C^+$ and $C^-$ is $2$.

The sequential and parallel compositions of pomsets can be lifted to pomset languages.
We also define a Kleene closure operator, similar to the one defined on languages of words.
\begin{definition}%
\label{definition:pomset-language-composition}
Let $\mathcal{U}$ and $\mathcal{V}$ be pomset languages.
We define:
\begin{align*}
\mathcal{U} \cdot \mathcal{V} &= \{ U \cdot V : U \in \mathcal{U}, V \in \mathcal{V} \} &
\mathcal{U} \parallel \mathcal{V} &= \{ U \parallel V : U \in \mathcal{U}, V \in \mathcal{V} \} &
\mathcal{U}^* &= \bigcup_{n \in \naturals} \mathcal{U}^n
\end{align*}
Where $\mathcal{U}^0 = \{ 1 \}$, and $\mathcal{U}^{n+1} = \mathcal{U} \cdot \mathcal{U}^n$ for all $n \in \naturals$.
\end{definition}
Kleene closure models indefinite repetition.
For instance, if our cookie recipe has a final step ``repeat until enough cookies have been made'', the pomset language $\mathcal{C}^*$ represents all possible traces of repetitions of the recipe; e.g., $C^+ \cdot C^+ \cdot C^- \in \mathcal{C}^*$ is the trace where first two batches of sprinkled cookies are made, followed by one without sprinkles.

\subsection{Series-rational expressions}

Just like a rational expression can be used to describe a regular structure of sequential events, a series-rational expression can be used to describe a regular structure of possibly parallel events.
Series-rational expressions are rational expressions with parallel composition.

\begin{definition}%
\label{definition:series-rational-expression}
The \emph{series-rational expressions}, denoted $\terms$, are formed by the grammar
\[e, f ::= 0 \pipe 1 \pipe a \in \Sigma \pipe e + f \pipe e \cdot f \pipe e \parallel f \pipe e^*\]
We use the symbols $d$, $e$, $f$, $g$ and $h$ to denote series-rational expressions.
\end{definition}

The semantics of a series-rational expression is given by a pomset language.

\begin{definition}%
\label{definition:series-rational-expression-semantics}
The function $\sem{-}: \terms \to 2^{\pom}$ is defined inductively, as follows:
\ifarxiv%
\begin{align*}
\sem{0} &= \emptyset \\
\sem{1} &= \{ 1 \} \\
\sem{a} &= \{ a \} \\
\sem{e + f} &= \sem{e} \cup \sem{f} \\
\sem{e \cdot f} &= \sem{e} \cdot \sem{f} \\
\sem{e \parallel f} &= \sem{e} \parallel \sem{f} \\
\sem{e^*} &= \sem{e}^*
\end{align*}
\else%
\begin{mathpar}
\sem{0} = \emptyset    \and
\sem{a} = \{ a \}      \and
\sem{1} = \{ 1 \}      \and
\sem{e^*} = \sem{e}^*  \\
\sem{e + f} = \sem{e} \cup \sem{f} \and
\sem{e \cdot f} = \sem{e} \cdot \sem{f} \and
\sem{e \parallel f} = \sem{e} \parallel \sem{f}
\end{mathpar}
\fi%
If $\mathcal{U} \in 2^{\pom}$ such that $\mathcal{U} = \sem{e}$ for some $e \in \terms$, then $\mathcal{U}$ is a \emph{series-rational language}.
\end{definition}

To illustrate, consider the pomset language $\mathcal{C}^* = {\{ C^+, C^- \}}^*$, which describes the possible traces arising from indefinitely repeating the cookie recipe, optionally adding chocolate sprinkles at every repetition.
We can describe the pomset language $\{ C^- \}$ with the series-rational expression $c^- = \mathsf{prepare} \cdot (\mathsf{bake} \parallel \mathsf{caramelize}) \cdot \mathsf{glaze}$, and $\{ C^+ \}$ by $c^+ = c^- \cdot \mathsf{sprinkle}$, which yields the series-rational expression $c = c^- + c^+$ for $\mathcal{C}$.
By construction, $\sem{c^*} = \mathcal{C}^*$.

\subsection{Additive congruence}

The following congruence on series-rational expressions will be instrumental in analyzing the automaton we introduce in Section~\ref{section:expressions-to-automata}, and for restricting said automaton to be finite in Section~\ref{subsection:bounding}.

\begin{definition}%
\label{definition:congruence}
We define $\simeq$ as the smallest congruence on $\terms$ such that:
\begin{mathpar}
 e_1 + 0 \simeq e_1
\and%
e_1 + e_1 \simeq e_1
\and%
e_1 + e_2 \simeq e_2 + e_1
\and%
e_1 + (e_2 + e_3) \simeq (e_1 + e_2) + e_3
\\
0 \cdot e_1 \simeq 0
\and%
e_1 \cdot 0 \simeq 0
\and%
0 \parallel e \simeq 0
\and%
e \parallel 0 \simeq 0
\end{mathpar}
When $\mset{g, h}, \mset{g', h'} \in \binom{\terms}{2}$ such that $g \simeq g'$ and $h \simeq h'$, we write $\mset{g, h} \simeq \mset{g', h'}$.
\end{definition}
Thus, when we claim that $e \simeq e'$, we say that $e$ is equal to $e'$, modulo associativity, commutativity and idempotence of $+$, as well as its unit $0$, and possibly annihilation of sequential and parallel composition by $0$.
Moreover, this congruence is sound with respect to the semantics, and it identifies all expressions that have an empty denotational semantics.

\begin{restatable}{lemma}{congruencesound}%
\label{lemma:congruence-sound}
Let $e, f \in \terms$.
If $e \simeq f$, then $\sem{e} = \sem{f}$.
Also, $e \simeq 0$ if and only if $\sem{e} = \emptyset$.
\end{restatable}
\begin{arxivproof}
Refer to Appendix~\ref{appendix:proofs}.
\end{arxivproof}
\noindent
There is a simple linear time decision procedure to test whether two expressions are congruent.
This justifies our using this relation to build finite automata later on.
As a by-product, we get that the emptiness problem for series-rational expressions is linear time decidable.
\section{Pomset Automata}%
\label{section:pomset-automata}

We are now ready to describe an automaton model that recognises series-rational languages.

\begin{definition}%
\label{definition:automaton}
A \emph{pomset automaton (PA)} is a tuple $\angl{Q, \delta, \gamma, F}$ where $Q$ is a set of \emph{states}, with $F \subseteq Q$ the \emph{accepting states},
$\delta: Q \times \Sigma \to Q$ is a function called the \emph{sequential transition function},
$\gamma: Q \times \binom{Q}{2} \to Q$ is a function called the \emph{parallel transition function}.
\end{definition}

Note that we do not fix an initial state.
As a result, a PA does not define a single pomset language but rather a mapping from its states to pomset languages.
The language of a state is defined in terms of a trace relation that involves the transitions of both $\delta$ and $\gamma$.
Here, $\delta$ plays the same role as in classic finite automata: given a state and a symbol, it returns the new state after reading that symbol.
The function $\gamma$ warrants a bit more explanation.
Given a state $q$ and a binary multiset of states $\mset{r, s}$, $\gamma$ tells us the state that is reached after reading two input streams in parallel starting at states $r$ and $s$, and having both ``subprocesses'' reach an accepting state.
The precise meaning is given in Definition~\ref{definition:automaton-language} below.

\begin{definition}%
\label{definition:automaton-language}
$\tracerel_A\ \subseteq Q \times \pom^+ \times Q$ is the smallest relation
\ifnotarxiv%
\footnote{The relation $\tracerel_A$ should not be thought of as deterministic; for fixed $q \in Q$ and $U \in \pom^+$, there may be multiple distinct $q' \in Q$ such that $q \atrace{U}_A q'$--- see the extended version~\cite{kappe-brunet-luttik-silva-zanasi-2017-arxiv} for additional information.}
\fi%
satisfying the rules
\begin{mathpar}
\inferrule{~}{%
    q \atrace{a}_A \delta(q, a)
}\and%
\inferrule{%
    q \atrace{U}_A q'' \\
    q'' \atrace{V}_A q'
}{%
    q \atrace{U \cdot V}_A q'
}\and%
\inferrule{%
    r \atrace{U}_A r' \in F \\
    s \atrace{V}_A s' \in F
}{%
    q \atrace{U \parallel V}_A \gamma(q, \mset{r, s})
}
\end{mathpar}
We also define $\rtracerel_A\ \subseteq Q \times \pom \times Q$ by $q \artrace{U}_A q'$ if and only if $q' = q$ and $U = 1$, or $q \atrace{U}_A q'$.
The \emph{language} of $A$ at $q \in Q$, denoted $L_A(q)$, is the set $\{ U : \exists q' \in F.\ q \artrace{U}_A q' \}$.
We say that $A$ \emph{accepts} the language $\mathcal{U}$ if there exists a $q \in Q$ such that $L_A(q) = \mathcal{U}$.
\end{definition}

Intuitively, $\gamma$ ensures that when a process forks at state $q$
into subprocesses starting at $r$ and $s$, if each of those reaches an
accepting state, then the processes can join at
$\gamma(q, \mset{r, s})$.

We purposefully omit the empty pomset $1$ as a label in $\tracerel_A$; doing so would open up the possibility of having traces of the form $q \atrace{1}_A q'$ with $q \neq q'$ (i.e., ``silent transitions'' or ``$\epsilon$-transitions'') for example by defining $\gamma(q, \mset{r, s}) = q'$ for some $r, s \in F$.
Avoiding transitions of this kind allows us to prove claims about $\tracerel_A$ by induction on the pomset size, and leverage Lemma~\ref{lemma:pomset-unique-factorization} in the process to disambiguate between the rules that apply.
By extension, we can prove claims about $\rtracerel_A$ and $L_A$ by treating $U = 1$ as a special case.

For the remainder of this section, we fix a PA $A = \angl{Q, \delta, \gamma, F}$, and a state $q\in Q$.
To simplify matters later on, we assume that $A$ has a state $\bot \in Q - F$ such that, for every $a \in \Sigma$, it holds that $\delta(\bot, a) = \bot$ and, for every $\phi \in \binom{Q}{2}$, it holds that $\gamma(\bot, \phi) = \bot$.
Such a \emph{sink state} is particularly useful when defining $\gamma$: for a fixed $q \in Q$ not all $\mset{r,s} \in \binom{Q}{2}$ may give a value of $\gamma(q, \mset{r,s})$ that contributes to the language accepted by $q$.
In such cases, we can define $\gamma(q, \mset{r,s}) = \bot$.
Alternatively, we could have allowed $\gamma$ to be a partial function; we chose $\gamma$ as a total function so as not to clutter the definition of derivatives in Section~\ref{section:expressions-to-automata}.

We draw a PA in a way similar to finite automata: each state (except $\bot$) is a vertex, and accepting states are marked by a double border.
To represent sequential transitions, we draw labelled edges; for instance, in Figure~\ref{figure:example-pa}, $\delta(q_0, \mathsf{prepare}) = q_1$.
To represent parallel transitions, we draw hyper-edges; for instance, in Figure~\ref{figure:example-pa}, $\gamma(q_1, \mset{q_3, q_4}) = q_2$.
To avoid clutter, we do not draw either of these edges types the target state is $\bot$.
It is not hard to verify that the pomset $C$ of the earlier example is accepted by the PA in Figure~\ref{figure:example-pa}.

In principle, the state space of a PA can be infinite; we use this in Section~\ref{section:expressions-to-automata} to define a PA that has all possible series-rational expressions as states.
It is however also useful to know when we can prune an infinite PA into a finite PA while preserving the languages of the retained states.
In Section~\ref{section:automata-to-expressions}, we use this to translate the PA to a series-rational expression.

Note that it is not sufficient to talk about reachable states, i.e., states that appear in the target of some trace; we must also include states that are ``meaningful'' starting points for subprocesses.
To do this, we first need a handle on these starting points.
Specifically, we are interested in the states where
\begin{inparaenum}[(1)]
    \item the eventual join of the states yields a state that contributes to the behaviour of the PA, and
    \item the states may join again, because they are not the sink state.
\end{inparaenum}
This is captured in the definition below.

\begin{definition}%
\label{definition:support}
The \emph{support} of $q$, written $\pi_A(q)$, is $\{ \mset{r, s} \in \binom{Q}{2} : \gamma(q, \mset{r, s}), r, s \neq \bot \}$.
\end{definition}

We can now talk about subsets of states of an automaton that are closed, in the sense that the relevant part of a transition function has input and output confined to this set.
As a result, we can confine the structure of a given PA to a closed set.

\begin{definition}%
\label{definition:closed}
A set of states $Q'\subseteq Q$ is \emph{closed} when the following rules are satisfied
\begin{mathpar}
\inferrule{~}{%
    \bot \in Q'
}\and%
\inferrule{%
    q \in Q' \\
    a \in \Sigma
}{%
    \delta(q, a) \in Q'
}\and%
\inferrule{%
    q \in Q' \\
    \phi \in \pi_A(q)
}{%
    \gamma(q, \phi) \in Q'
}\and%
\inferrule{%
    q \in Q' \\
    \mset{r, s} \in \pi_A(q)
}{%
    r, s \in Q'
}
\end{mathpar}
If $Q'$ is closed, the \emph{generated sub-PA} of $A$ induced by $Q'$, denoted $\restr{A}{Q'}$, is the tuple $\angl{Q', \restr{\delta}{Q'}, \restr{\gamma}{Q'}, Q' \cap F}$
where $\restr{\delta}{Q'}$ and $\restr{\gamma}{Q'}$ are the restrictions of $\delta$ and $\gamma$ to $Q'$.
\end{definition}

Because the relevant parts of the transition functions are preserved, it is not surprising that the language of a state in a generated sub-PA coincides with the language of that state in the original PA\@.

\begin{lemma}%
\label{lemma:subpa-preserves-language}
Let $Q' \subseteq Q$ be closed.
If $q \in Q'$, then
$L_{\restr{A}{Q'}}(q) = L_A(q)$.
\end{lemma}
\begin{arxivproof}
First, observe that $1 \in L_{A'}(q)$ if and only if $1 \in L_A(q)$, since $q \in Q' \cap F$ if and only if $q \in F$.
Now, suppose that $q \atrace{U}_{A'} q'$; we show that in this case $q \atrace{U}_A q'$ follows.
The proof proceeds by induction on $U$.
In the base, $U = a \in \Sigma$, and $q' = \restr{\delta}{Q'}(q, a)$.
But then $q' = \delta(q, a)$, and so $q \atrace{U}_A q'$.
For the inductive step, there are two cases to consider:
\begin{itemize}
    \item If $U = V \cdot W$ with $V$ and $W$ smaller than $U$, there exists a $q'' \in Q'$ such that $q \atrace{V}_{A'} q''$ and $q'' \atrace{W}_{A'} q'$.
    By induction, we find $q \atrace{V}_A q''$ and $q'' \atrace{W}_A q'$, and thus $q \atrace{U}_A q'$.
    \item If $U = V \parallel W$ with $V$ and $W$ smaller than $U$, then there exist $r, s \in Q'$ and $r', s' \in F \cap Q'$ such that $r \atrace{V}_{A'} r'$ and $s \atrace{W}_{A'} s'$, and $q' = \restr{\gamma}{Q'}(q, \mset{r, s})$.
    By induction we find that $r \atrace{V}_A r'$ and $s \atrace{W}_A s'$; note that $r', s' \in F$.
    It then follows that $q \atrace{U}_A \gamma(q, \mset{r, s}) = q'$.
\end{itemize}

For the other direction, suppose that $q \atrace{U}_A q'$ and $q \neq \bot$; we show that in this case $q \atrace{U}_{A'} q'$ follows.
The proof proceeds by induction on $U$.
In the base, $U = a \in \Sigma$ and $q' = \delta(q, a)$.
But then $q' = \restr{\delta}{Q'}(q, a)$ and so $q \atrace{U}_{A'} q'$.
For the inductive step, there are two cases to consider:
\begin{itemize}
    \item If $U = V \cdot W$ with $V$ and $W$ smaller than $U$, there exists a $q'' \in Q$ such that $q \atrace{V}_A q''$ and $q'' \atrace{W}_A q'$.
    Since $q' \neq \bot$, we have that $q'' \neq \bot$.
    By induction we then find that $q \atrace{V}_{A'} q''$ (thus $q'' \in Q'$) and so again by induction $q'' \atrace{W}_{A'} q'$.
    In total, we obtain $q \atrace{U}_{A'} q'$.
    \item If $U = V \parallel W$ with $V$ and $W$ smaller than $U$, there exist $r, s \in Q$ and $r', s' \in F$ such that $r \atrace{V}_A r'$ and $s \atrace{W}_A s'$, and $q' = \gamma(q, \mset{r, s})$.
    Since $q' \neq \bot$, it holds that $\mset{r, s} \in \pi_A(q)$, thus $r, s \in Q'$; also, since $r', s' \in F$, we know that $r', s' \neq \bot$.
    By induction we find that $r \atrace{V}_{A'} r'$ and $s \atrace{W}_{A'} s'$.
    Note that, since $r', s' \in F$, also $r, s \neq \bot$, and therefore $\mset{r, s} \in \pi_A(q)$.
    We then conclude that $q \atrace{U}_{A'} \restr{\gamma}{Q'}(q, \mset{r, s}) = \gamma(q, \mset{r, s}) = q'$. \qedhere
\end{itemize}
\end{arxivproof}

We now work out how to find a closed subset of states that contains a particular state.
The first step is to characterize the states reachable from $q$ by means of transitions.

\begin{definition}%
\label{definition:reach}
The \emph{reach} of $q$, written $\rho_A(q)$, is the smallest set
satisfying the rules
\begin{mathpar}
\inferrule{~}{%
    q \in \rho_A(q)
}\and%
\inferrule{%
    q' \in \rho_A(q) \\
    a \in \Sigma
}{%
    \delta(q', a) \in \rho_A(q)
}\and%
\inferrule{%
    q' \in \rho_A(q) \\
    \phi \in \pi_A(q)
}{%
    \gamma(q', \phi) \in \rho_A(q)
}
\end{mathpar}
\end{definition}

The reach of a state is closely connected to the states that can be reached from $q$ through the trace relation of the automaton, in the following way:
\begin{restatable}{lemma}{reachvsreachable}%
\label{lemma:reach-vs-reachable} The set $\rho_A(q) \cup \{ \bot \}$ contains $\{ q' \in Q : \exists U \in \pom^+.\ q \atrace{U}_A q' \} \cup \{ q \}$.
\ifarxiv%
Moreover, if $\bot$ is the only state of $A$ whose language is empty, this containment is an equality.
\fi%
\end{restatable}
\begin{arxivproof}
Refer to Appendix~\ref{appendix:proofs}.
\end{arxivproof}

Note that $\rho_A(q) \cup \{ \bot \}$ is not necessarily closed: we also need the states required by the fourth rule of closure in Definition~\ref{definition:closed}.
Thus, if we want to ``close'' $\rho_A(q) \cup \{ \bot \}$ by adding the support of its contents, we need to find closed sets of states that contain branching points.
In order to do this inductively, we propose the following subclass of PAs.

\begin{definition}%
\label{definition:fork-acyclic}
We say that $A$ is \emph{fork-acyclic} if there exists a \emph{fork
  hierarchy}, which is a strict order $\prec_A\ \subseteq Q \times Q$
such that the following rules are satisfied.
\begin{mathpar}
\inferrule{%
    \mset{r, s} \in \pi_A(q)
}{%
    r, s \prec_A q
}\and%
\inferrule{%
    a \in \Sigma \\
    r \prec_A \delta(q, a)
}{%
    r \prec_A q
}\and%
\inferrule{%
    \phi \in \pi_A(q) \\
    r \prec_A \gamma(q, \phi)
}{%
    r \prec_A q
}
\end{mathpar}
\end{definition}

The fork hierarchy is connected with the reach of a state in the following way.
\begin{lemma}%
\label{lemma:reach-vs-hierarchy}
Let $q', r \in Q$.
If $A$ is fork-acyclic, $q' \in \rho_A(q)$ and $r \prec_A q'$, then $r \prec_A q$.
\end{lemma}
\begin{arxivproof}
The proof proceeds by induction on the construction of $\rho_A(q)$.
In the base, $q = q'$, in which case the claim holds vacuously.

For the inductive step, there are two cases to consider.
\begin{itemize}
    \item If $q' = \delta(q'', a)$ for some $q'' \in \rho_A(q)$ and $a \in \Sigma$, then we know that $r \prec_A q''$ by definition of $\prec_A$.
    But then $r \prec_A q$ by induction.
    \item If $q' = \gamma(q'', \phi)$ for some $q'' \in \rho_A(q)$ and $\phi \in \pi_A(q'')$, then we know that $r \prec_A q''$ by definition of $\prec_A$.
    But then $r \prec_A q$ by induction. \qedhere
\end{itemize}
\end{arxivproof}

The term \emph{fork-acyclic} has been used in literature for similar automata~\cite{lodaya-weil-2000,jipsen-2014}.
However, in op.\ cit., it is defined in terms of the traces that arise from the transition structure of the automaton.
In contrast, our definition is purely syntactic: it imposes an order on states such that forks cannot be nested.
To show that, as in~\cite{lodaya-weil-2000}, our definition implies that languages of the PA have bounded width, we present the following lemma.
Since the state space of a PA can be infinite, we additionally require that the fork hierarchy is well-founded.

\begin{lemma}
  If $A$ is fork-acyclic and $\prec_A$ is well-founded then $L_A(q)$
  is of finite width.
\end{lemma}
\begin{arxivproof}
If $q \in Q$, let $n_q \in \naturals$ be the lowest upper bound on the length of a descending chain starting at $q$, i.e., if $r_1, r_2, \dots, r_m \in Q$ are such that $r_1 \prec_A r_2 \prec_A \dots \prec_A r_m \prec_A q$, then $m \leq n_q$.
Such an $n_q$ exists uniquely, for $\prec_A$ is well-founded.
We strengthen our claim as follows: if $q \in Q$, then $\width{L_A(q)} \leq 2^{n_q}$.

It suffices to show that if $q' \in Q$ and $U \in \pom^+$ are such that $q \atrace{U}_A q'$ and $q' \neq \bot$, then $\width{U} \leq 2^{n_q}$.
The proof proceeds by nested induction.
The outer induction is on $n_q$; here, we assume the claim holds for all $r \in Q$ with $n_r < n_q$ (note that this implicitly covers the base, where $n_q = 0$).
The inner induction is on $U$.
In the base of the inner induction, $U = a \in \Sigma$ and the claim holds immediately, for $\width{a} = 1 \leq 2^{n_q}$.

For the inductive step of the inner induction, there are two cases to consider.
\begin{itemize}
    \item If $U = V \cdot W$, with $V$ and $W$ smaller than $U$, then there exists a $q'' \in Q$ such that $q \atrace{V}_A q''$ and $q'' \atrace{W}_A q'$.
    Since $q' \neq \bot$, we know that $q'' \neq \bot$ as well.
    Therefore, we know by the (inner and outer) induction hypothesis that $\width{V} \leq 2^{n_q}$.
    Furthermore, since $q'' \in \rho_A(q)$ by Lemma~\ref{lemma:reach-vs-reachable} (the case where $q'' = \bot$ was excluded), we know that $q'' \prec_A q$ by Lemma~\ref{lemma:reach-vs-hierarchy}, therefore $n_{q''} < n_q$.
    Thus, by the (inner) induction hypothesis, it follows that that $\width{W} \leq 2^{n_{q''}}$.
    But then $\width{U} = \max(\width{V}, \width{W}) \leq \max(2^{n_q}, 2^{n_{q''}}) = 2^{n_q}$ as well.
    \item If $U = V \parallel W$, with $V$ and $W$ smaller than $U$, then there exist $r, s \in Q$ and $r', s' \in F$ such that $r \atrace{V}_A r'$ and $s \atrace{W}_A s'$.
    By the outer induction hypothesis, we know that $\width{V} \leq \width{L_A(r)} \leq 2^{n_r}$ and $\width{W} \leq \width{L_A(s)} \leq 2^{n_s}$.
    Furthermore, since $q', r, s \neq \bot$, we know that $\mset{r, s} \in \pi_A(q)$, thus $r, s \prec_A q$ and therefore $n_s, n_r < n_q$.
    We then derive
    \[\width{U} = \width{V} + \width{W} \leq 2^{n_r} + 2^{n_s} \leq 2 \cdot 2^{\max(n_r, n_s)} = 2^{\max(n_r, n_s) + 1} \leq 2^{n_q}\]
    This completes the proof. \qedhere
\end{itemize}
\end{arxivproof}

We introduce the notion of a \emph{bounded PA}, which is sufficient to guarantee the existence of a closed, finite subset containing a given state, even when the PA has infinitely many states.

\begin{definition}%
\label{definition:bounded}
Let $A$ be fork-acyclic.
We say that $A$ is \emph{bounded} if $\prec_A$
is well-founded, and for all $q \in Q$, both $\pi_A(q)$ and
$\rho_A(q)$ are finite.
\end{definition}

\begin{theorem}%
\label{theorem:bounded-finite-subpa}
If $A$ is bounded, then for every state $q$ of $A$ there exists a finite set of states
$\project{Q}{q} \subseteq Q$ that is closed and contains $q$.
\end{theorem}
\begin{proof}
The proof proceeds by $\prec_A$-induction; this is sound, because $\prec_A$ is well-founded.

Suppose the claim holds for all $r \in Q$ with $r \prec_A q$.
If $q' \in \rho_A(q)$ and $\mset{r, s} \in \pi_A(q')$, then $r \prec_A q'$ and thus $r \prec_A q$ by Lemma~\ref{lemma:reach-vs-hierarchy}; by induction we obtain for every such $r$ a finite set of states $\project{Q}{r} \subseteq Q$ that is closed and contains $r$.
We choose:
\[\project{Q}{q} = \{ \bot \} \cup \rho_A(q) \cup \bigcup \{ \project{Q}{r} : q' \in \rho_A(q),\ \mset{r, s} \in \pi_A(q') \}\]
This set is finite because $\rho_A(q)$ and $\pi_A(q')$ are finite for all $q, q' \in Q$ since $A$ is bounded.
To see that $\project{Q}{q}$ is closed, it suffices to show that the last rule of closure holds for $q' \in \rho_A(q)$; it does, since if $q' \in \rho_A(q)$ and $\mset{r, s} \in \pi_A(q')$, then $r \in \project{Q}{r}$ and $s \in \project{Q}{s}$, thus $r, s \in \project{Q}{q}$.
\end{proof}

\ifarxiv%
\subsection{Lock-step traces}

Note that while the transition functions of a PA are deterministic in that their output is a single state rather than a set of states, the transition relation of a PA $A = \angl{Q, \delta, \gamma, F}$ should not be thought of as deterministic.
In particular, if $q \atrace{U}_A q'$ and $q \atrace{U}_A q''$, then $q' = q''$ does not hold in general.
For instance, consider the case where $\gamma(q, \mset{r_1, s_1}) = q'$ and $\gamma(q, \mset{r_2, s_2}) = q''$ and furthermore $s_1 \atrace{V}_A s_1' \in F$ and $s_2 \atrace{V}_A s_2' \in F$ as well as $r_1 \atrace{W}_A r_1' \in F$ and $r_2 \atrace{W}_A r_2' \in F$.
Here, $q \atrace{V \parallel W}_A q'$ and $q \atrace{V \parallel W}_A q''$, but $q'$ and $q''$ may not be equal.

When working with traces in a PA, it is may be useful to be able to relate traces of the same pomset in the presence of this kind of determinism.
To this end, we introduce the notion of lock-step traces.
Intuitively, lock-step traces prevent the counterexample discussed above by requiring that the application of the third rule must use the same starting states in the construction of both traces.
\begin{definition}%
\label{definition:lock-step}
Let $A = \angl{Q, \delta, \gamma, F}$ be a PA and let $U \in \pom^+$.
Suppose that $q_1 \atrace{U}_A q_1'$ and $q_2 \atrace{U}_A q_2'$; these traces are \emph{in lock-step} if one of the following is true:
\begin{enumerate}[(i)]
    \item $U = a$ for some $a \in \Sigma$
    \item $U = V \cdot W$, for $V$ and $W$ smaller than $U$, and there exist $q_1'', q_2'' \in Q$ such that $q_1 \atrace{V}_A q_1''$ and $q_2 \atrace{V}_A q_2''$ are in lock step, as well as $q_1'' \atrace{W} q_1'$ and $q_2'' \atrace{W} q_2'$
    \item $U = V \parallel W$, for $V$ and $W$ smaller than $U$, and there exist $r, s \in Q$ and $r', s' \in F$ such that $r \atrace{V} r'$ and $s \atrace{W} s'$, as well as $q_1' = \gamma(q_1, \mset{r, s})$ and $q_2' = \gamma(q_2, \mset{r, s})$.
\end{enumerate}
\end{definition}
It is easy to see that ``being in lock-step'' is an equivalence relation on traces.

The lemma below observes that traces that are in lock-step do enjoy determinism.
\begin{lemma}%
\label{lemma:lock-step-determinism}
Let $A = \angl{Q, \delta, \gamma, F}$ be a PA and let $q, q_1', q_2' \in Q$ and $U \in \pom^+$.
If $q \atrace{U}_A q_1'$ and $q \atrace{U}_A q_2'$ are in lock-step, then $q_1' = q_2'$.
\end{lemma}
\begin{proof}
The proof proceeds by induction on $U$.
If $U = a \in \Sigma$, then $q_1' = \delta(q, a)$ and $q_2' = \delta(q, a)$, and so the claim follows.

For the inductive step, there are two cases to consider.
\begin{itemize}
    \item If $U = V \cdot W$ with $V$ and $W$ smaller than $U$, then there exist $q_1'', q_2'' \in Q$ such that $q \atrace{V}_A q_1''$ and $q \atrace{V}_A q_2''$ are in lock-step, as well as $q_1'' \atrace{W}_A q_1'$ and $q_2'' \atrace{W}_A q_2'$.
    By induction, $q_1'' = q_2''$, and thus again by induction $q_1' = q_2'$.
    \item If $U = V \parallel W$ with $V$ and $W$ smaller than $U$, then there exist $r, s \in Q$ and $r', s' \in F$ such that $r \atrace{V} r'$ and $S \atrace{W} s'$, and $q_1' = \gamma(q, \mset{r, s})$ and $q_2' = \gamma(q, \mset{r, s})$.
    The claim then follows. \qedhere
\end{itemize}
\end{proof}
\fi%

\section{Expressions to automata}%
\label{section:expressions-to-automata}

We now turn our attention to the task of translating a series-rational expression $e$ into a PA that accepts $\sem{e}$.
We employ Brzozowski's method~\cite{brzozowski-1964} to construct a single \emph{syntactic PA} where every series-rational expression is a state accepting exactly its denotational semantics.
To this end we must define which expressions are accepting, and how the sequential and parallel transition functions transform states --- what are, in Brzozowski's vocabulary, their sequential and parallel derivatives?

We start with the accepting states.
In Brzozowski's construction, a rational expression is accepting if its denotational semantics includes the empty word.
Analogously, a series-rational expression is accepting if its denotational semantics includes the empty pomset.

\begin{definition}%
\label{definition:accepting-terms}
We define the set $\sacc$ to be the smallest subset of $\terms$ satisfying the rules:
\begin{mathpar}
\inferrule{~}{%
    1 \in \sacc
}\and%
\inferrule{%
    e \in \sacc \\
    f \in \terms
}{%
    e + f, f + e \in \sacc
}\and%
\inferrule{%
    e, f \in \sacc
}{%
    e \cdot f, f \cdot e \in \sacc
}\and%
\inferrule{%
    e, f \in \sacc
}{%
    e \parallel f, f \parallel e \in \sacc
}\and%
\inferrule{%
    e \in \terms
}{%
    e^* \in \sacc
}
\end{mathpar}
\end{definition}
It is not hard to see that $e \in \sacc$ if and only if $1 \in \sem{e}$.
We use $e \star f$ as a shorthand for $f$ if $e \in \sacc$, and $0$ otherwise.
For an equation $\mathcal{E}$, we write $[\mathcal{E}]$ as a shorthand for $1$ if $\mathcal{E}$ holds, and $0$ otherwise.
We now define sequential and parallel derivatives:
\begin{definition}%
\label{definition:derivatives}
\ifarxiv%
We define the functions $\ssderiv: \terms \times \Sigma \to \terms$ and $\psderiv: \terms \times \binom{\terms}{2} \to \terms$ as follows:
\begin{align*}
\ssderiv(0, a) &= 0                                                             & \psderiv(0, \phi) &= 0 \\
\ssderiv(1, a) &= 0                                                             & \psderiv(1, \phi) &= 0 \\
\ssderiv(b, a) &= [a = b]                                                       & \psderiv(b, \phi) &= 0 \\
\ssderiv(e + f, a) &= \ssderiv(e, a) + \ssderiv(f, a)                           & \psderiv(e + f, \phi) &= \psderiv(e, \phi) + \psderiv(f, \phi) \\
\ssderiv(e \cdot f, a) &= \ssderiv(e, a) \cdot f + e \star \ssderiv(f, a)       & \psderiv(e \cdot f, \phi) &= \psderiv(e, \phi) \cdot f + e \star \psderiv(f, \phi) \\
\ssderiv(e \parallel f, a) &= e \star \ssderiv(f, a) + f \star \ssderiv(e, a)   & \psderiv(e \parallel f, \phi) &= [\phi \simeq \mset{e, f}] + e \star \psderiv(f, \phi) + f \star \psderiv(e, \phi) \\
\ssderiv(e^*, a) &= \ssderiv(e, a) \cdot e^*                                    & \psderiv(e^*, \phi) &= \psderiv(e, \phi) \cdot e^*
\end{align*}
\else%
We define the function $\ssderiv: \terms \times \Sigma \to \terms$ as follows:
\begin{mathpar}
\ssderiv(0, a) = 0        \and
\ssderiv(1, a) = 0        \and
\ssderiv(b, a) = [a = b]  \and
\ssderiv(e^*, a) = \ssderiv(e, a) \cdot e^*\and
\ssderiv(e + f, a) = \ssderiv(e, a) + \ssderiv(f, a) \and
\ssderiv(e \cdot f, a) = \ssderiv(e, a) \cdot f + e \star \ssderiv(f, a) \and
\ssderiv(e \parallel f, a) = e \star \ssderiv(f, a) + f \star \ssderiv(e, a)
\end{mathpar}
Furthermore, the function $\psderiv: \terms \times \binom{\terms}{2} \to \terms$ is defined as follows:
\begin{mathpar}
\psderiv(0, \phi) = 0 \and
\psderiv(1, \phi) = 0 \and
\psderiv(b, \phi) = 0 \and
\psderiv(e^*, \phi) = \psderiv(e, \phi) \cdot e^*\\
\psderiv(e + f, \phi) = \psderiv(e, \phi) + \psderiv(f, \phi)\and
\psderiv(e \cdot f, \phi) = \psderiv(e, \phi) \cdot f + e \star \psderiv(f, \phi)\and
\psderiv(e \parallel f, \phi) = [\phi \simeq \mset{e, f}] + e \star \psderiv(f, \phi) + f \star \psderiv(e, \phi)
\end{mathpar}
\fi%
\end{definition}
The definition of $\ssderiv$ coincides with Brzozowski's derivative on rational expressions.
The definition of $\psderiv$ mimics the definition of $\ssderiv$ on non-parallel terms except $b \in \Sigma$.

The definition of $\psderiv$ on parallel terms includes (in the first term) the possibility that the starting states provided to the parallel transition function are (congruent to) the operands of the parallel, in which case the target join state is the accepting state $1$.
The other two terms (as well as the definition of $\ssderiv$ on a parallel term) account for the fact that if $1 \in \sem{e}$, then $\sem{f} \subseteq \sem{e \parallel f}$.
Since we do not allow traces labelled with the empty pomset, traces that originate from these operands are thus lifted to the composition when necessary.

\begin{definition}
The \emph{syntactic PA} is the PA $A_\Sigma = \angl{\terms, \ssderiv, \psderiv, \sacc}$.
\end{definition}
We use $\slang$ as a shorthand for $L_{A_\Sigma}$, and $\satracerel$ ($\sartracerel$) as a shorthand for $\tracerel_{A_\Sigma}$ ($\rtracerel_{A_\Sigma}$).

The remainder of this section is devoted to showing that if $e \in \terms$, then $\slang(e) = \sem{e}$.

\subsection{Traces of congruent states}

In the analysis of the syntactic trace relation $\satracerel$, we often encounter sums of terms.
To work with these, it is useful to identify terms modulo $\simeq$.
In this section, we establish that such an identification is in fact sound, in the sense that if two expressions are related by $\simeq$, then the languages accepted by the states representing those expressions are also identical.

In the first step towards this goal, we show that $\sacc$ is well-defined with respect to $\simeq$.
\begin{restatable}{lemma}{congruentstateshalt}%
\label{lemma:congruent-states-halt}
Let $e, f \in \terms$ be such that $e \simeq f$.
Then $e \in \sacc$ if and only if $f \in \sacc$.
\end{restatable}
\begin{arxivproof}
Refer to Appendix~\ref{appendix:proofs}.
\end{arxivproof}

Also, $\ssderiv$ and $\psderiv$ are well-defined with respect to $\simeq$, in the following sense:
\begin{restatable}{lemma}{congruentstatesderive}%
\label{lemma:congruent-states-derive}
Let $e, f \in \terms$ such that $e \simeq f$.
If $a \in \Sigma$, then $\ssderiv(e, a) \simeq \ssderiv(f, a)$.
Moreover, if $\phi = \mset{g, h} \in \binom{\terms}{2}$ with $g, h \not\simeq 0$, then $\psderiv(e, \phi) \simeq \psderiv(f, \phi)$, and if $\psi \in \binom{\terms}{2}$ with $\phi \simeq \psi$, then $\psderiv(e, \phi) = \psderiv(e, \psi)$.
\end{restatable}
\begin{arxivproof}
Refer to Appendix~\ref{appendix:proofs}.
\end{arxivproof}

With these lemmas in hand, we can show that $\simeq$ is a ``bisimulation'' with respect to $\satracerel$.
\begin{lemma}%
\label{lemma:congruent-states-progress}
Let $e, f \in \terms$ be such that $e \simeq f$.
If $e \satrace{U} e'$, then there exists an $f' \in \terms$ such that $f \satrace{U} f'$ and $e' \simeq f'$.
\end{lemma}
\begin{arxivproof}
The proof proceeds by induction on $U$.
In the base, $U = a \in \Sigma$.
But then $e' = \ssderiv(e, a)$.
If we choose $f' = \ssderiv(f, a)$, we find that $f \satrace{U} f'$ and $e' \simeq f'$ by Lemma~\ref{lemma:congruent-states-derive}.

For the inductive step, there are two cases to consider.
\begin{itemize}
    \item If $U = V \cdot W$ with $V$ and $W$ smaller than $U$, there exists an $e'' \in \terms$ such that $e \satrace{V} e''$ and $e'' \satrace{W} e'$.
    By induction, we find an $f'' \in \terms$ such that $f \satrace{V} f''$ and $e'' \simeq f''$.
    Again by induction, we find $f' \in \terms$ such that $f'' \satrace{W} f'$ and $e' \simeq f'$.
    This means that $f \satrace{U} f'$ with $e' \simeq f'$.
    \item If $U = V \parallel W$ with $V$ and $W$ smaller than $U$, there exist $g, h \in \terms$ and $g', h' \in \sacc$ such that $g \satrace{V} g'$ and $h \satrace{W} h'$, and $e' = \psderiv(e, \mset{g, h})$.
    Note that, in this case, $g, h \neq 0$ by Lemma~\ref{lemma:trace-deconstruct-base} (which will be proved shortly).
    If we choose $f' = \psderiv(f, \mset{g, h})$ we find that $f \satrace{U} f'$ with $e' \simeq f'$ by Lemma~\ref{lemma:congruent-states-derive}. \qedhere
\end{itemize}
\end{arxivproof}

Let $I$ be a finite set, and let ${(e_i)}_{i \in I}$ be an $I$-indexed family of terms.
In the sequel, we treat $\sum_{i \in I} e_i$ as a term, where the $e_i$ are summed in some arbitrary order or bracketing.
The lemmas above guarantee that the precise choice of representing this sum as a term makes no matter with regard to the traces allowed.

\subsection{Trace deconstruction}%
\label{subsection:trace-deconstruction}

We proceed with a series of lemmas that characterise reachable states in the syntactic PA\@.
More precisely, we show that the expressions reachable from some expression $e$ can be written as sums of expressions reachable from subexpressions of $e$.
For this reason, we refer to these observations as \emph{trace deconstruction lemmas}: they deconstruct a trace of an expression into traces of ``smaller'' expressions.
The purpose of these lemmas is twofold; in Section~\ref{subsection:soundness}, they are used to characterise the languages of expressions as they appear in the syntactic PA, while in Section~\ref{subsection:bounding} they allow us to bound the reach of an expression.

We start by analysing the traces that originate in base terms, such as $0$, $1$, or $a \in \Sigma$.
\begin{lemma}%
\label{lemma:trace-deconstruct-base}
Let $e, e' \in \terms$ and $U \in \pom^+$ such that $e \satrace{U} e'$.
If $e \in \{ 0, 1 \}$, then $e' = 0$.
Furthermore, if $e = b \in \Sigma$, then either $e' = 1$ and $U = b$, or $e' = 0$.
\end{lemma}
\begin{arxivproof}
The proof of the first claim proceeds by induction on $U$.
In the base, $U = a \in \Sigma$, and $e' = \ssderiv(e, a) = 0$, thus the claim holds.

For the inductive step, there are two cases to consider.
\begin{itemize}
    \item If $U = V \cdot W$, with $V$ and $W$ smaller than $U$, then there exists an $e'' \in \terms$ such that $e \satrace{V} e''$ and $e'' \satrace{W} e'$.
    By induction, $e'' = 0$, and thus again by induction $e' = 0$.
    \item If $U = V \parallel W$, with $V$ and $W$ smaller than $U$, then there exist $f, g \in \terms$ and $f', g' \in \sacc$ such that $f \satrace{V} f'$ and $g \satrace{W} g'$, and furthermore $e' = \psderiv(e, \mset{f, g}) = 0$.
\end{itemize}

The proof of the second claim also proceeds by induction on $U$.
In the base, $U = a \in \Sigma$.
Then $e' = \ssderiv(b, a)$.
If $a = b$, then $e' = 1$, otherwise $e' = 0$; thus the claim holds.

For the inductive step, there are two cases to consider.
\begin{itemize}
    \item If $U = V \cdot W$ with $V$ and $W$ smaller than $U$, then there exists an $e'' \in \terms$ such that $b \satrace{V} e''$ and $e'' \satrace{W} e'$.
    By induction, either $V = b$ and $e'' = 1$, or $e'' = 0$.
    In either case, $e' = 0$ by the first claim.
    \item If $U = V \parallel W$ with $V$ and $W$ smaller than $U$, then there exist $f, g \in \terms$ and $f', g' \in \sacc$ such that $f \satrace{V} f'$ and $g' \satrace{W} g'$, and $e' = \psderiv(e, \mset{f, g}) = 0$.
    The claim thus holds immediately. \qedhere
\end{itemize}
\end{arxivproof}
Note, however, that $0$ and $1$ are not indistinguishable, for $0 \not\in \sacc$ while $1 \in \sacc$.

We also consider the traces that originate in a sum of terms.
The intuition here is that the input is processed by both terms simultaneously, and thus the target state must be the sum of the states that are the result of processing the input for each term individually.
\begin{lemma}%
\label{lemma:trace-decompose-plus}
Let $e_1, e_2 \in \terms$ and $U \in \pom^+$.
If $e_1 + e_2 \satrace{U} e'$, then there exist $e_1', e_2' \in \terms$ such that $e' = e_1' + e_2'$, and $e_1 \satrace{U} e_1'$ and $e_2 \satrace{U} e_2'$.
\end{lemma}
\begin{arxivproof}
The proof proceeds by induction on $U$.
In the base, $U = a \in \Sigma$, and $e' = \ssderiv(e_1 + e_2, a) = \ssderiv(e_1, a) + \ssderiv(e_2, a)$.
We can then choose $e_1' = \ssderiv(e_1, a)$ and $e_2' = \ssderiv(e_2, a)$ to validate the claim.

For the inductive step, there are two cases to consider.
\begin{itemize}
    \item If $U = V \cdot W$ with $V$ and $W$ smaller than $U$, then there exists an $e'' \in \terms$ such that $e_1 + e_2 \satrace{V} e''$ and $e'' \satrace{W} e'$.
    By induction, we find $e_1'', e_2'' \in \terms$ such that $e'' = e_1'' + e_2''$, and $e_1 \satrace{V} e_1''$ and $e_2 \satrace{V} e_2''$.
    Again by induction, we find $e_1', e_2' \in \terms$ such that $e' = e_1' + e_2'$, and $e_1'' \satrace{W} e_1'$ and $e_2'' \satrace{W} e_2'$.
    In conclusion, we find that $e_1 \satrace{U} e_1'$ and $e_2 \satrace{U} e_2'$.
    \item If $U = V \parallel W$, with $V$ and $W$ smaller than $U$, then there exist $f, g \in \terms$ such that $e' = \psderiv(e_1 + e_2, \mset{f, g}) = \psderiv(e_1, \mset{f, g}) + \psderiv(e_2, \mset{f, g})$, and furthermore there exist $f', g' \in \sacc$ such that $f \satrace{V} f'$ and $g \satrace{W} g'$.
    We can choose $e_1' = \psderiv(e_1, \mset{f, g})$ and $e_2' = \psderiv(e_2, \mset{f, g})$ to validate the claim. \qedhere
\end{itemize}
\end{arxivproof}

We now consider the traces starting in a sequential composition.
The intuition here is that the syntactic PA must first proceed through the left operand, before it can proceed to process the right operand.
Thus, either the pomset is processed by the left operand entirely, or we should be able to split the pomset in two sequential parts: the first part is processed by the left operand, and the second by the right operand.
\begin{lemma}%
\label{lemma:trace-deconstruct-sequential}
Let $e_1, e_2 \in \terms$ and $U \in \pom^+$ be such that $e_1 \cdot e_2 \satrace{U} f$.
There exist an $f' \in \terms$ and a finite set $I$, as well as $I$-indexed families ${(f_i')}_{i \in I}$ over $\sacc$ and ${(f_i)}_{i \in I}$ over $\terms$, and $I$-indexed families ${(U_i')}_{i \in I}$, ${(U_i)}_{i \in I}$ over $\pom^+$, such that:
\begin{itemize}
    \item $f \simeq f' \cdot e_2 + \sum_{i \in I} f_i$ and $e_1 \satrace{U} f'$, and
    \item for all $i \in I$, $e_1 \sartrace{U_i'} f_i'$, $e_2 \satrace{U_i} f_i$, and $U = U_i' \cdot U_i$.
\end{itemize}
\end{lemma}
\begin{arxivproof}
The proof proceeds by induction on $U$.
In the base, $U = a \in \Sigma$, and $f = \ssderiv(e_1 \cdot e_2, a)$.
If on the one hand $e_1 \not\in \sacc$, then $f \simeq \ssderiv(e_1, a) \cdot e_2$; we choose $f' = \ssderiv(e_1, a)$ and $I = \emptyset$.
If on the other hand $e_1 \in \sacc$, then $f \simeq \ssderiv(e_1, a) \cdot e_2 + \ssderiv(e_2, a)$; we choose $f' = \ssderiv(e_1, a)$, $I = \{ * \}$, $f_*' = e_1$, $f_* = \ssderiv(e_2, a)$, $U_*' = 1$ and $U_* = U$.
In either case, our choices validate the claim.

For the inductive step, we consider two cases.
\begin{itemize}
    \item If $U = V \cdot W$, with $V$ and $W$ smaller than $U$, then there exists a $g$ such that $e_1 \cdot e_2 \satrace{V} g$ and $g \satrace{W} f$.
    By induction, we obtain a $g' \in \terms$ and a finite set $J$, as well as $J$-indexed families ${(g_j')}_{j \in J}$ over $\sacc$ and ${(g_j)}_{j \in J}$ over $\terms$, and $J$-indexed families ${(V_j')}_{j \in J}$ and ${(V_j)}_{j \in J}$ over $\pom^+$, such that
    \begin{itemize}
        \item $g \simeq g' \cdot e_2 + \sum_{j \in J} g_j$ and $e_1 \satrace{V} g'$, and
        \item for all $j \in J$, $e_1 \sartrace{V_j'} g_j'$, $e_2 \satrace{V_j} g_j$ and $V = V_j' \cdot V_j$.
    \end{itemize}
    By Lemma~\ref{lemma:congruent-states-progress} and $g \simeq g' \cdot e_2 + \sum_{j \in J} g_i$, as well as Lemma~\ref{lemma:trace-decompose-plus} and $g \satrace{W} f$, there exist $h \in \terms$ and a $J$-indexed family of terms ${(d_j)}_{j \in J}$ over $\terms$ such that $f \simeq h + \sum_{j \in J} d_j$, $g' \cdot e_2 \satrace{W} h$ and for all $j \in J$, also $g_j \satrace{W} d_j$.

    Since $g' \cdot e_2 \satrace{W} h$, we can find (again by induction) an $h' \in \terms$, and a finite set $K$ (without loss of generality, disjoint from $J$) and $K$-indexed families ${(h_k')}_{k \in K}$ over $\sacc$ and ${(h_k)}_{k \in K}$ over $\terms$, as well as $K$-indexed families ${(W_k')}_{k \in K}$ over $\pom$ and ${(W_k)}_{k \in K}$ over $\pom^+$ such that
    \begin{itemize}
        \item $h \simeq h' \cdot e_2 + \sum_{k \in K} h_k$ and $g' \satrace{W} h'$, and
        \item for all $k \in K$, $g' \sartrace{W_k'} h_k'$, $e_2 \satrace{W_k} h_k$ and $W = W_k' \cdot W_k$.
    \end{itemize}

    We now choose $I = J \cup K$, and $f' = h'$.
    Furthermore, we choose $I$-indexed families ${(f_i')}_{i \in I}$ over $\sacc$ and ${(f_i)}_{i \in I}$ over $\terms$, as well as $I$-indexed families ${(U_i')}_{i \in I}$ over $\pom$ and ${(U_i)}_{i \in I}$ over $\pom^+$ as follows:
    \begin{align*}
    f_i' &=
    \begin{cases}
    g_i' & i \in J \\
    h_i' & i \in K
    \end{cases}
    &
    f_i &=
    \begin{cases}
    d_i & i \in J \\
    h_i & i \in K
    \end{cases}
    &
    U_i' &=
    \begin{cases}
    V_i' & i \in J \\
    V \cdot W_i' & i \in K
    \end{cases}
    &
    U_i &=
    \begin{cases}
    V_i \cdot W & i \in J \\
    W_i & i \in K
    \end{cases}
    \end{align*}

    It remains to verify the requirements on our choices one by one.
    For the first claim, we can derive
    \[f \simeq h + \sum_{j \in J} d_j \simeq h' \cdot e_2 + \sum_{k \in K} h_k + \sum_{j \in J} d_j \simeq f' \cdot e_2 + \sum_{i \in I} f_i\]
    Furthermore, since $e_1 \satrace{V} g'$ and $g' \satrace{W} h'$, we have that $e_1 \satrace{U} h' = f'$.
    For the second claim, let us fix an $i \in I$.
    Suppose first that $i \in J$, then:
    \begin{itemize}
        \item $e_1 \sartrace{V_i'} g_i'$, thus $e_1 \sartrace{U_i'} f_i'$
        \item $e_2 \satrace{V_i} g_i$ and $g_i \satrace{W} d_i$, thus $e_2 \satrace{U_i} f_i$
        \item $U_i' \cdot U_i = V_i' \cdot V_i \cdot W = V \cdot W = U$.
    \end{itemize}
    Secondly, suppose that $i \in K$, then:
    \begin{itemize}
        \item $e_1 \satrace{V} g'$ and $g' \sartrace{W_i'} h_i'$, thus $e_1 \sartrace{U_i'} f_i'$
        \item $e_2 \satrace{W_i} h_i$, thus $e_2 \satrace{U_i} f_i$
        \item $U_i' \cdot U_i = V \cdot W_i' \cdot W_i = V \cdot W = U$.
    \end{itemize}
    All requirements are thus validated for our choices in this case.

    \item If $U = V \parallel W$, then there exist $g, h \in \terms$ and $g', h' \in \sacc$ such that $g \satrace{V} g'$ and $h \satrace{W} h'$, and we know that $f = \psderiv(e_1 \cdot e_2, \mset{g, h})$.
    If on the one hand $e_1 \not\in \sacc$, then $f \simeq \psderiv(e_1, \mset{g, h}) \cdot e_2$; we choose $f' = \psderiv(e_1, \mset{g_1, g_2})$ and $I = \emptyset$.
    If on the other hand $e_1 \in \sacc$, then $f \simeq \psderiv(e_1, \mset{g, h}) \cdot e_2 + \psderiv(e_2, \mset{g, h})$; we choose $f' = \psderiv(e_1, \mset{g_1, g_2})$ and $I = 1$, as well as $f_*' = e_1$, $f_* = \psderiv(e_2, \mset{g_1, g_2})$, $U_*' = 1$ and $U_*' = U$.
    In both cases, our choices validate the claim. \qedhere
\end{itemize}
\end{arxivproof}

The next deconstruction lemma concerns traces originating in a parallel composition.
Intuitively, the syntactic PA either processes parallel components of the pomset, or processes according to one operand, provided that the other operand allows immediate acceptance.

\begin{lemma}%
\label{lemma:trace-deconstruct-parallel}
If $e_1 \parallel e_2 \satrace{U} f$, then there exist $f_1, f_2, f_3 \in \terms$, such that
\begin{itemize}
    \item $f \simeq f_1 + f_2 + f_3$,
    \item either $f_1 = 0$, or $e_2 \in \sacc$ and $e_1 \satrace{U} f_1$,
    \item either $f_2 = 0$, or $e_1 \in \sacc$ and $e_2 \satrace{U} f_2$, and
    \item either $f_3 = 0$, or $f_3 = 1$ and there exist $f_1', f_2' \in \sacc$ and $U_1, U_2 \in \pom^+$ such that $U = U_1 \parallel U_2$ and $e_1 \satrace{U_1} f_1'$ and $e_2 \satrace{U_2} f_2'$.
\end{itemize}
\end{lemma}
\begin{arxivproof}
The proof proceeds by induction on $U$.
In the base, $U = a \in \Sigma$.
Now, choose $f_2 = e_1 \star \ssderiv(e_2, a)$ and $f_1 = e_2 \star \ssderiv(e_1, a)$.
We also choose $f_3 = 0$.
It is easy to validate that the claim holds for these choices.

For the inductive step, there are two cases to consider.
\begin{itemize}
    \item If $U = V \cdot W$ with $V$ and $W$ smaller than $U$, then there exists a $g \in \terms$ such that $e_1 \parallel e_2 \satrace{V} g$ and $g \satrace{W} f$.
    By induction, we obtain $g_1, g_2, g_3$ such that:
    \begin{itemize}
        \item $g \simeq g_1 + g_2 + g_3$,
        \item either $g_1 = 0$, or $e_2 \in \sacc$ and $e_1 \satrace{V} g_1$,
        \item either $g_2 = 0$, or $e_1 \in \sacc$ and $e_2 \satrace{V} g_2$, and
        \item either $g_3 = 0$, or $g_3 = 1$ and there exist $e_1', e_2' \in \terms$ and $g_1', g_2' \in \sacc$ and $V_1, V_2 \in \pom^+$ such that $e_1 \simeq e_1'$ and $e_2 \simeq e_2'$ and $V = V_1 \parallel V_2$ and $e_1' \satrace{V_1} g_1'$ and $e_2' \satrace{V_2} g_2'$.
    \end{itemize}
    By $g \simeq g_1 + g_2 + g_3$ and Lemma~\ref{lemma:congruent-states-progress}, as well as $g \satrace{W} f$ and Lemma~\ref{lemma:trace-decompose-plus}, we obtain $f_1, f_2, f_3 \in \terms$ such that $f \simeq f_1 + f_2 + f_3$, and $g_1 \satrace{W} f_1$, $g_2 \satrace{W} f_2$ and $g_3 \satrace{W} f_3$.
    We now validate the remaining claims.
    First, note that if $g_i = 0$ for all $i \in \{ 1, 2, 3 \}$, then $f_i = 0$ by Lemma~\ref{lemma:trace-deconstruct-base}.
    As to the remaining possibilities:
    \begin{itemize}
        \item If $g_1 \neq 0$, then $e_2 \in \sacc$ and by $e_1 \satrace{V} g_1$ and $g_1 \satrace{W} f_1$ we conclude that $e_1 \satrace{U} f_1$.
        \item If $g_2 \neq 0$, then $e_1 \in \sacc$ and by $e_2 \satrace{V} g_2$ and $g_2 \satrace{W} f_2$ we conclude that $e_2 \satrace{U} f_2$.
        \item If $g_3 \neq 0$, then $g_3 = 1$ and we conclude that $f_3 = 0$ by Lemma~\ref{lemma:trace-deconstruct-base}.
    \end{itemize}

    \item If $U = V \parallel W$ with $V$ and $W$ smaller than $U$, then there exist $g, h \in \terms$ and $g', h' \in \sacc$ such that $f = \psderiv(e_1 \parallel e_2, \mset{g, h})$, and $g \satrace{V} g' \in \sacc$ and $h \satrace{W} h' \in \sacc$.
    We choose $f_1 = e_2 \star \psderiv(e_1, \mset{g, h})$ and $f_2 = e_1 \star \psderiv(e_2, \mset{g, h})$.
    Furthermore, we set $f_3 = [\mset{e_1, e_2} = \mset{g, h}]$.
    It is now easy to see that $f \simeq f_1 + f_2 + f_3$.
    We validate the remaining claims.
    \begin{itemize}
        \item If $e_2 \not\in \sacc$, then $f_1 = 0$.
        Otherwise, $e_2 \in \sacc$ and $e_1 \satrace{U} \psderiv(e_1, \mset{g, h}) = f_1$.
        \item If $e_1 \not\in \sacc$, then $f_2 = 0$.
        Otherwise, $e_1 \in \sacc$ and $e_2 \satrace{U} \psderiv(e_2, \mset{g, h}) = f_2$.
        \item If $\mset{e_1, e_2} \not\simeq \mset{g, h}$, then $f_3 = 0$.
        Otherwise, assume (without loss of generality) that $e_1 \simeq g$ and $e_2 \simeq h$.
        By Lemma~\ref{lemma:congruent-states-progress} and the fact that $g \satrace{V} g'$ as well as $h \satrace{W} h'$, there exist $f_1', f_2' \in \terms$ such that $e_1 \satrace{V} f_1'$ and $e_2 \satrace{W} f_2'$, with $f_1' \simeq g'$ and $f_2' \simeq h'$.
        By Lemma~\ref{lemma:congruent-states-halt} and the fact that $g', h' \in \sacc$, it then follows that $f_1, f_2' \in \sacc$.
        Choosing $U_1 = V$ and $U_2 = W$ now validates the claim. \qedhere
    \end{itemize}
\end{itemize}
\end{arxivproof}

Finally, we analyse the reachable states of an expression of the form $e^*$.
The intuition here is that, starting in $e^*$, the PA can iterate traces originating in $e$ indefinitely.
The trace should thus be sequentially decomposable, with each component the label of a trace originating in $e$.
Furthermore, all but the last target state of these traces should be accepting.

\begin{lemma}%
\label{lemma:trace-deconstruct-star}
If $e^* \satrace{U} f$, then there exists a finite set $I$ and an $I$-indexed family of finite sets ${(J_i)}_{i \in I}$, as well as $I$-indexed families ${(f_i)}_{i \in I}$ over $\terms$ and ${(U_i)}_{i \in I}$ over $\pom^+$, and for all $i \in I$ also $J_i$-indexed families ${(f_{i,j})}_{j \in J_i}$ over $\sacc$ and ${(U_{i,j})}_{j \in J_i}$ over $\pom^+$, such that $f \simeq \sum_{i \in I} f_i \cdot e^*$, and for all $i \in I$:
\begin{itemize}
    \item $e \satrace{U_i} f_i$,
    \item for all $j \in J_i$ we have that $e \satrace{U_{i,j}} f_{i,j}$, and
    \item $U = U_i' \cdot U_i$, where $U_i'$ is some concatenation of all $U_{i,j}$ for all $j \in J_i$.
\end{itemize}
\end{lemma}
\begin{arxivproof}
The proof proceeds by induction on $U$.
In the base, $U = a \in \Sigma$, and so $f = \ssderiv(e^*, a) = \ssderiv(e, a) \cdot e^*$.
We can choose $I = \{ * \}$, $J_* = \emptyset$ and $U_* = a$ to validate the claim.

For the inductive step, we consider two cases.
\begin{itemize}
    \item If $U = V \cdot W$, with $V$ and $W$ smaller than $U$, there exists a $g$ such that $e^* \satrace{V} g$ and $g \satrace{W} f$.
    Since the remainder of this part of the proof is somewhat involved, we begin by outlining our strategy.
    First, we deconstruct the trace $e^* \satrace{V} g$, by induction.
    Then we deconstruct the trace $g \satrace{W} f$, making use of the fact that $g$ can be seen a sum of the form found in the claim, i.e., where each term is of the form $g' \cdot e^*$, and thus (by Lemma~\ref{lemma:trace-decompose-plus}) $f$ can be seen as a sum of terms $f'$ such that $g' \cdot e^* \satrace{W} f'$.
    We then leverage Lemma~\ref{lemma:trace-deconstruct-sequential} to deconstruct each of the latter traces.
    In the end, we have a big cache of variables, which we use to construct the families of terms and pomsets required by the claim.
    Since we will be dealing with a fair number of index sets, we tacitly assume (without loss of generality) that of them are disjoint.

    As for the proof, consider the trace $e^* \satrace{V} g$.
    By induction, we obtain a finite set $I'$ and an $I'$-indexed family of finite sets ${(J_i')}_{i \in I'}$, as well as $I'$-indexed families ${(g_i)}_{i \in I'}$ over $\terms$ and ${(V_i)}_{i \in I'}$ over $\pom^+$ and for all $i \in I'$ also $J_i'$-indexed families ${(g_{i,j})}_{j \in J_i'}$ over $\sacc$ and ${(V_{i,j})}_{j \in J_i'}$ over $\pom^+$, such that $g \simeq \sum_{i \in I'} g_i \cdot e^*$, and for all $i \in I'$:
    \begin{itemize}
        \item $e \satrace{V_i} g_i$,
        \item for all $j \in J_i'$ we have that $e \satrace{V_{i,j}} g_{i,j} \in \sacc$, and
        \item $V = V_i' \cdot V_i$, where $V_i'$ is some concatenation of $V_{i,j}$ for all $j \in J_i'$.
    \end{itemize}
    By the fact that $g \simeq \sum_{i \in I'} g_i \cdot e^*$ and by Lemma~\ref{lemma:congruent-states-progress}, as well as $g \satrace{W} f$ and Lemma~\ref{lemma:trace-decompose-plus}, we find an $I'$-indexed family ${(h_i)}_{i \in I'}$ over $\terms$ such that $f \simeq \sum_{i \in I'} h_i$, and for all $i \in I'$ also $g_i \cdot e^* \satrace{W} h_i$.
    Then, by Lemma~\ref{lemma:trace-deconstruct-sequential}, we find for all $i \in I'$ a term $h_i'$ and finite set $K_i$ as well as $K_i$-indexed families ${(h_{i,k}')}_{k \in K_i}$ over $\sacc$ and ${(h_{i,k})}_{k \in K_i}$ over $\terms$ and ${(W_{i,k}')}_{k \in K_i}$ over $\pom$ and ${(W_{i,k})}_{k \in K_i}$ over $\pom^+$, such that
    \begin{itemize}
        \item $h_i \simeq h_i' \cdot e^* + \sum_{k \in K_i} h_{i,k}$ and $g_i \satrace{W} h_i'$, and
        \item for all $k \in K_i$, $g_i \sartrace{W_i'} h_{i,k}'$, $e^* \satrace{W_i} h_{i,k}$ and $W = W_k' \cdot W_k$.
    \end{itemize}

    By induction and since for all $i \in I'$ and $k \in K_i$ we have that $e^* \satrace{W_i} h_{i,k}$, we obtain for this $i$ and $k$ a finite set $L_{i,k}$ and an $L_{i,k}$-indexed family of finite sets ${(M_{i,k,\ell})}_{\ell \in L_{i,k}}$, as well as $L_{i,k}$-indexed families ${(h_{i,k,\ell})}_{\ell \in L_{i,k}}$ over $\terms$ and ${(W_{i,k,\ell})}_{\ell \in L_{i,k}}$ over $\pom^+$, and for all $\ell \in L_{i,k}$ also $M_{i,k,\ell}$-indexed families ${(h_{i,k,\ell,m})}_{m \in M_{i,k,\ell}}$ over $\sacc$ and ${(V_{i,k,\ell,m})}_{m \in M_{i,k,\ell}}$ over $\pom^+$ such that $h_{i,k} \simeq \sum_{\ell \in L_{i,k}} h_{i,k,\ell} \cdot e^*$, and for all $\ell \in L_{i, k}$:
    \begin{itemize}
        \item $e \satrace{W_{i,k,\ell}} h_{i,k,\ell}$,
        \item for all $m \in M_{i,k,\ell}$ we have that $e \satrace{W_{i,k,\ell,m}} h_{i,k,\ell,m}$, and
        \item $W_i = W_{i,k,\ell}' \cdot W_{i,k,\ell}$, where $W_{i,k,\ell}'$ is some concatenation of $W_{i,k,\ell,m}$ for $m \in M_{i,k,\ell}$.
    \end{itemize}
    We are now ready to choose the required (families of) sets, terms and pomsets, as follows:
    \begin{itemize}
        \item $I = I' + \bigcup_{k \in K_i} L_{i,k}$
        \item for all $i \in I$, we set
        \begin{align*}
        J_i &=
        \begin{cases}
        J_i' & i \in I' \\
        M_{i',k,i} \cup \{ *_i \} & i \in L_{i',k}, k \in K_{i'}
        \end{cases}
        &
        f_i &=
        \begin{cases}
        h_i' & i \in I' \\
        h_{i',k,i} & i \in L_{i',k}, k \in K_{i'}
        \end{cases}
        \intertext{%
        where $*_i$ is a ``fresh'' symbol not in any index set.
        We furthermore choose
        }
        U_i &=
        \begin{cases}
        V_i \cdot W & i \in I' \\
        W_{i',k,i} & i \in L_{i',k}, k \in K_{i'}
        \end{cases}
        \end{align*}
        \item for all $i \in I$ and $j \in J_i$, we set
        \begin{align*}
        f_{i,j} &=
        \begin{cases}
        g_{i,j} & i \in I', j \in J_i' \\
        h_{i',k}' & i \in L_{i',k}, k \in K_{i'}, j = *_i \\
        h_{i',k,i,j} & i \in L_{i',k}, k \in K_{i'}, j \neq *_i
        \end{cases}
        \intertext{and furthermore}
        U_{i,j} &=
        \begin{cases}
        V_{i,j} & i \in I', j \in J_i' \\
        V_{i'} \cdot W_{i'}' & i \in L_{i',k}, k \in K_{i'}, j = *_i \\
        W_{i',k,i,j} & i \in L_{i',k}, k \in K_{i'}, j \neq *_i
        \end{cases}
        \end{align*}
    \end{itemize}
    It remains to check the requirements on our choices.
    \begin{itemize}
        \item One easily verifies that
        \[f \simeq \sum_{i \in I'} h_i \simeq \sum_{i \in I'} \left( h_i' \cdot e^* + \sum_{i \in K_i} h_{i,k} \right) \simeq \sum_{i \in I'} \left( h_i' \cdot e^* + \sum_{k \in K_i} \sum_{\ell \in L_{i,k}} h_{i,k,\ell} \cdot e^* \right) \simeq \sum_{i \in I} f_i \cdot e^*\]
        \item If $i \in I$, there are two cases to consider.
        \begin{itemize}
            \item If $i \in I'$, then since $e \satrace{V_i} g_i$ and $g_i \satrace{W} h_i'$, we find that $e \satrace{U_i} f_i$.
            \item If $i \in L_{i',k}$ for some $i' \in I'$ and $k \in K_{i'}$, then since $e \satrace{W_{i',k,i}} h_{i',k,i}$ we find that $e \satrace{U_i} f_i$.
        \end{itemize}
        \item If $i \in I$ and $j \in J_i$, there are three cases to consider.
        \begin{itemize}
            \item If $i \in I'$ and $j \in J_i'$, then since $e \satrace{V_{i,j}} g_{i,j}$, we find that $e \satrace{U_{i,j}} f_{i,j}$.
            \item If $i \in L_{i',k}$ for some $i' \in I'$ and $k \in K_{i'}$, and $j = *_i$, then since $e \satrace{V_{i'}} g_i$ and $g_i \sartrace{W_{i'}'} h_{i',k}'$, we find that $e \satrace{U_{i,j}} f_{i,j}$.
            \item If $i \in L_{i',k}$ for some $i' \in I'$ and $k \in K_{i'}$, and $j \neq *_i$ (thus $j \in M_{i',k,i}$), then since $e \satrace{W_{i',k,i,j}} h_{i',k,i,j}$ we find that $e \satrace{U_{i,j}} f_{i,j}$.
        \end{itemize}
        \item If $i \in I$, there are two cases to consider.
        \begin{itemize}
            \item If $i \in I'$, then $V = V_i' \cdot V_i$, where $V_i'$ is some concatenation of $V_{i,j}$ for all $j \in J_i'$.
            But then we can choose $U_i' = V_i'$ as a concatenation of $U_{i,j}$ for all $j \in J_i$, to find that $U_i' \cdot U_i = V_i' \cdot V_i \cdot W = V \cdot W = U$.
            \item If $i \in L_{i',k}$ for all $i' \in I'$ and $k \in K_{i'}$, then we know that there exists a pomset $W_{i',k,i}'$ which is a concatenation of $W_{i',k,i,m}$ for all $m \in M_{i',k,i}$, such that $W_{i',k,i}' \cdot W_{i',k,i} = W_{i'}$.
            We can then choose $U_i' = V \cdot W_{i'}' \cdot W_{i',k,i}'$ as a concatenation of $U_{i,j}$ with $j \in J_i$, and find that $U_i' \cdot U_i = V \cdot W_{i'}' \cdot W_{i',k,i}' \cdot W_{i',k,i} = V \cdot W_{i'}' \cdot W_{i'} = V \cdot W = U$.
        \end{itemize}
    \end{itemize}

    \item If $U = V \parallel W$, then there exist $g, h \in \terms$ and $g', h' \in \terms$ such that $g \satrace{V} g'$ and $h \satrace{W} h'$, and $f = \psderiv(e^*, \mset{g, h}) = \psderiv(e, \mset{g, h}) \cdot e^*$.
    In this case, we choose $I = \{ * \}$, and $J_* = \emptyset$ and $U_* = U$ to find that the claims are validated. \qedhere
\end{itemize}
\end{arxivproof}

\subsection{Trace construction}%
\label{subsection:trace-construction}

In the above, we learned how to deconstruct traces in the syntactic PA\@.
To verify that the state in the syntactic PA associated with a series-rational expression $e$ indeed accepts the series-rational pomset language $\sem{e}$, we also need to show the converse, that is, how to \emph{construct} traces in the syntactic PA from smaller traces.
In this context it is often useful to work with the preorder obtained from $\simeq$.
\begin{definition}
  The relation $\lesssim\ \subseteq \terms \times \terms$ is defined
  by $e \lesssim f$ if and only if $e + f \simeq f$.
\end{definition}
The intuition to $e \lesssim f$ is that $e$ consists of one or more terms that also appear in $f$, up to $\simeq$.

In analogy to Lemma~\ref{lemma:congruent-states-progress}, we show that $\lesssim$ is a ``simulation'' with respect to traces.
\begin{lemma}%
\label{lemma:subcongruent-states-simulation}
Let $e, e', f \in \terms$ be such that $e \lesssim f$.
If $e \satrace{U} e'$, then there exists an $f' \in \terms$ such that $f \satrace{U} f'$ and $e' \lesssim f'$.
Furthermore, if $e \in \sacc$, then $f \in \sacc$.
\end{lemma}
\begin{arxivproof}
We prove the first claim by induction on $U$.
In the base, $U = a \in \Sigma$ and $e' = \ssderiv(e, a)$.
Note that $\ssderiv(e, a) + \ssderiv(f, a) = \ssderiv(e + f, a) \simeq \ssderiv(f, a)$ by Lemma~\ref{lemma:congruent-states-derive}.
We choose $f' = \ssderiv(f, a)$ to find that $f \satrace{U} f'$ with $e' = \ssderiv(e, a) \lesssim \ssderiv(f, a) = f'$.

For the inductive step, there are two cases to consider.
\begin{itemize}
    \item If $U = V \cdot W$ with $V$ and $W$ smaller than $U$, then there exists an $e''$ such that $e \satrace{V} e''$ and $e'' \satrace{W} e'$.
    By induction, we find an $f'' \in \terms$ such that $f \satrace{V} f''$ and $e'' \lesssim f''$, and again by induction we find an $f' \in \terms$ such that $f'' \satrace{W} f'$ with $e' \lesssim f'$.
    In total, we know that $e \satrace{U} f$ with $e' \lesssim f'$.
    \item If $U = V \parallel W$ with $V$ and $W$ smaller than $U$, then there exist $g, h \in \terms$ and $g', h' \in \sacc$ such that $g \satrace{V} g'$ and $h \satrace{W} h'$ and $e' = \psderiv(e, \mset{g, h})$.
    Note that $\psderiv(e, \mset{g, h}) + \psderiv(f, \mset{g, h}) = \psderiv(e + f, \mset{g, h}) \simeq \psderiv(f, \mset{g, h})$ by Lemma~\ref{lemma:congruent-states-derive}.
    We choose $f' = \psderiv(f, \mset{g, h})$ to find that $f \satrace{U} f'$ with $e' = \psderiv(e, \phi) \lesssim \psderiv(f, \phi) = f'$.
\end{itemize}

For the second claim, suppose that $e \in \sacc$, then also $e + f \in \sacc$.
But then $f \in \sacc$ by Lemma~\ref{lemma:congruent-states-halt}.
\end{arxivproof}

The following lemma tells us that we can create a trace labelled with the concatenation of the labels of two smaller traces, and starting in the sequential composition of the original starting states, provided that the first trace ends in an accepting state.
Furthermore, the target state of the newly constructed trace contains the target state of the second trace.
\ifarxiv%
We also prove two auxiliary claims towards this end, which will be useful later on.
\fi%
\begin{lemma}%
\label{lemma:trace-construct-sequential}
\ifarxiv%
Let $e_1, e_2, f_1, f_2 \in \terms$ and $U, V \in \pom^+$ be such that $e_1 \satrace{U} f_1$ and $e_2 \satrace{V} f_2$.
The following hold:
\begin{itemize}
    \item There exists an $f \in \terms$ such that $e_1 \cdot e_2 \satrace{U} f$ with $f_1 \cdot e_2 \lesssim f$.
    \item If $f_1 \in \sacc$, then there exists an $f \in \terms$ such that $f_1 \cdot e_2 \satrace{V} f$ with $f_2 \lesssim f$.
    \item If $f_1 \in \sacc$, then there exists an $f \in \terms$ such that $e_1 \cdot e_2 \satrace{U \cdot V} f$ with $f_2 \lesssim f$.
\end{itemize}
\else%
Let $e_1, e_2, f_2 \in \terms$ and $f_1 \in \sacc$.
If $U, V \in \pom^+$ are such that $e_1 \satrace{U} f_1$ and $e_2 \satrace{V} f_2$, then there exists an $f \in \terms$ such that $e_1 \cdot e_2 \satrace{U \cdot V} f$ with $f_2 \lesssim f$.
\fi%
\end{lemma}
\begin{arxivproof}
The proof of the first claim proceeds by induction on $U$.
In the base, $U = a \in \Sigma$ and $f_1 = \ssderiv(e_1, a)$.
We choose $f = \ssderiv(e_1 \cdot e_2, a)$.
It is now easy to show that $e_1 \cdot e_2 \satrace{U} f$ with $f_1 \cdot e_2 \lesssim f$.

For the inductive step, there are two cases to consider.
\begin{itemize}
    \item If $U = W \cdot X$ with $W$ and $X$ smaller than $U$, there exists a $g_1$ such that $e_1 \satrace{W} g_1$ and $g_1 \satrace{X} f_1$.
    By induction, we obtain $g \in \terms$ such that $e_1 \cdot e_2 \satrace{W} g$ with $g_1 \cdot e_2 \lesssim g$.
    Again by induction, we find $h \in \terms$ such that $g_1 \cdot e_2 \satrace{X} h$ with $f_1 \cdot e_2 \lesssim h$.
    By Lemma~\ref{lemma:subcongruent-states-simulation}, we then know that $g \satrace{W} f$ for some $f$ with $h \lesssim f$.
    In total, we find that $e_1 \cdot e_2 \satrace{U} f$ with $f_1 \cdot e_2 \lesssim f$ (by transitivity of $\lesssim$).
    \item If $U = W \parallel X$ with $W$ and $X$ smaller than $U$, there exist $g, h \in \terms$ and $g', h' \in \sacc$ such that $g \satrace{W} g'$ and $h \satrace{X} h'$, and $f_1 = \psderiv(e_1, \mset{g, h})$.
    We choose $f = \psderiv(e_1 \cdot e_2, \mset{g, h})$.
    It is now easy to show that $e_1 \cdot e_2 \satrace{U} f$ with $f_1 \cdot e_2 \lesssim f$.
\end{itemize}

The proof of the second claim proceeds by induction on $V$.
In the base, $V = a \in \Sigma$ and $f_2 = \ssderiv(e_2, a)$.
We can then choose $f = \ssderiv(f_1 \cdot e_2, a)$.
It is now easy to show that $f_1 \cdot e_2 \satrace{V} f$ with $f_2 \lesssim f$.

For the inductive step, there are two cases to consider.
\begin{itemize}
    \item If $V = W \cdot X$ with $W$ and $X$ smaller than $V$, there exists a $g_2 \in \terms$ such that $e_2 \satrace{W} g_2$ and $g_2 \satrace{X} f_2$.
    By induction, there exists a $g \in \terms$ such that $f_1 \cdot e_2 \satrace{W} g$ with $g_2 \lesssim g$.
    By Lemma~\ref{lemma:subcongruent-states-simulation}, we find $f$ such that $g \satrace{X} f$ with $f_2 \lesssim f$, and thus $f_1 \cdot e_2 \satrace{V} f$.
    \item If $V = W \parallel X$ with $W$ and $X$ smaller than $V$, there exist $g, h \in \terms$ and $g', h' \in \sacc$ such that $g \satrace{W} g'$ and $h \satrace{X} h'$ and $f_2 = \psderiv(e_2, \mset{g, h})$.
    We can then choose $f = \psderiv(f_1 \cdot e_2, \mset{g, h})$.
    It is now easy to show that $f_1 \cdot e_2 \satrace{V} f$ with $f_2 \lesssim f$.
\end{itemize}

The third claim is a direct consequence of the first two claims and Lemma~\ref{lemma:subcongruent-states-simulation}.
\end{arxivproof}

We can also construct traces that start in a parallel composition.
One way is to construct traces that start in each operand and reach an accepting state; we obtain a trace in their parallel composition almost trivially.
If one of the operands is accepting, we can also construct a single trace that starts in the other operand and obtain a trace with the same label starting in the parallel construction.
In both cases, we describe the target of the new trace using $\lesssim$.
\begin{lemma}%
\label{lemma:trace-construct-parallel}
Let $e_1, e_2 \in \terms$.
The following hold:
\begin{itemize}
    \item If $f_1, f_2 \in \sacc$ and $U, V \in \pom^+$ are such that $e_1 \satrace{U} f_1$ and $e_2 \satrace{V} f_2$, then there exists an $f \in \terms$ such that $e_1 \parallel e_2 \satrace{U \parallel V} f$ with $1 \lesssim f$.
    \item If $e_2 \in \sacc$ (respectively $e_1 \in \sacc$), and $f' \in \terms$ and $U \in \pom^+$ are such that $e_1 \satrace{U} f'$ (respectively $e_2 \satrace{U} f'$), then there exists an $f \in \terms$ such that $e_1 \parallel e_2 \satrace{U} f$ with $f' \lesssim f$.
\end{itemize}
\end{lemma}
\begin{arxivproof}
For the first claim, choose $f = \psderiv(e_1 \parallel e_2, \mset{e_1, e_2})$.
We then immediately find that $e_1 \parallel e_2 \satrace{U \parallel V} f$ with $1 \lesssim f$.

For the second claim, the proof proceeds by induction on $U$.
In the base, $U = a \in \Sigma$ and $f_1 = \ssderiv(e_1, a)$.
Choose $f = \ssderiv(e_1 \parallel e_2, a)$.
It is then easy to see that $e_1 \parallel e_2 \satrace{U} f$ with $f_1 \lesssim f$.

In the inductive step, there are two cases to consider.
\begin{itemize}
    \item If $U = V \cdot W$ with $V$ and $W$ smaller than $U$, there exists a $g_1$ such that $e_1 \satrace{V} g_1$ and $g_1 \satrace{W} f_1$.
    By induction, we find $g \in \terms$ such that $e_1 \parallel e_2 \satrace{V} g$ with $g_1 \lesssim g$.
    By Lemma~\ref{lemma:subcongruent-states-simulation}, we find $f \in \terms$ such that $g \satrace{W} f$ and $f_1 \lesssim f$.
    In total, we have that $e_1 \parallel e_2 \satrace{U} f$ with $f_1 \lesssim f$.
    \item If $U = V \parallel W$ with $V$ and $W$ smaller than $U$, there exist $g, h \in \terms$ and $g', h' \in \sacc$ such that $g \satrace{V} g'$ and $h \satrace{W} h'$, and $f_1 = \psderiv(e_1, \mset{g, h})$.
    We choose $f = \psderiv(e_1 \parallel e_2, \mset{g, h})$.
    It is now easy to see that $e_1 \parallel e_2 \satrace{U} f$ with $f_1 \lesssim f$. \qedhere
\end{itemize}
\end{arxivproof}

Lastly, we present a trace construction lemma to obtain traces originating in expressions of the form $e^*$.
The idea here is that, given a finite number of traces that originate in $e$, where all (but possibly one) have an accepting state as their target, we can construct a trace originating in $e^*$, with a concatenation of the labels of the input traces as its label.
\begin{lemma}%
\label{lemma:trace-construct-star}
Let $e, f_1, f_2, \dots, f_n \in \terms$ (with $n > 0$) be such that $f_1, f_2, \dots, f_{n-1} \in \sacc$.
Also, let $U, U_1, U_2, \dots, U_n \in \pom^+$ be such that $U = U_1 \cdot U_2 \cdots U_n$.
If for all $i \leq n$ it holds that $e \satrace{U_i} f_i$, then there exists an $f \in \terms$ such that $e^* \satrace{U} f$, with $f_n \cdot e^* \lesssim f$.
\end{lemma}
\begin{arxivproof}
First, we show that the claim holds for $n = 1$, i.e., if $e, f_1 \in \terms$ and $U \in \pom^+$ are such that $e \satrace{U} f_1$, then there exists an $f \in \terms$ such that $e^* \satrace{U} f$, with $f_1 \cdot e^* \lesssim f$.
The proof proceeds by induction on $U$.
In the base, $U = a \in \Sigma$ and $f_1 = \ssderiv(e, a)$.
We choose $f = \ssderiv(e^*, a) = \ssderiv(e, a) \cdot e^* = f_1 \cdot e^*$ to find that $e^* \satrace{U} f$ with $f_1 \cdot e^* \lesssim f$.

For the inductive step, there are two cases to consider.
\begin{itemize}
    \item If $U = V \cdot W$ with $V$ and $W$ smaller than $U$, there exists a $g_1 \in \terms$ such that $e \satrace{V} g_1$ and $g_1 \satrace{W} f_1$.
    By induction, we find $g \in \terms$ such that $e^* \satrace{V} g$ with $g_1 \cdot e^* \lesssim g$.
    By Lemma~\ref{lemma:trace-construct-sequential}, we find $h \in \terms$ such that $g_1 \cdot e^* \satrace{W} h$, with $f_1 \cdot e^* \lesssim h$.
    By Lemma~\ref{lemma:subcongruent-states-simulation}, we find $f$ such that $g \satrace{W} f$ with $h \lesssim f$.
    In total, we have that $e^* \satrace{U} f$ with $f_1 \cdot e^* \lesssim f$.
    \item If $U = V \parallel W$ with $V$ and $W$ smaller than $U$, there exist $g, h \in \terms$ and $g', h' \in \sacc$ such that $g \satrace{V} g'$, and $h \satrace{W} h'$, and $f_1 = \psderiv(e, \mset{g, h})$.
    In this case, we choose $f = \psderiv(e^*, \mset{g, h})$ to find that $e^* \satrace{U} f$ with $f_1 \cdot e^* \lesssim f$.
\end{itemize}

We now inductively extend the claim to all $n > 0$, using the proof above as our base.
In the inductive step, we assume the claim holds for $n-1$, and try to prove it for $n$.
Let $U' = U_1 \cdot U_2 \cdots U_{n-1}$.
By induction, we obtain $g \in \terms$ such that $e^* \satrace{U'} g$ with $f_{n-1} \cdot e^* \lesssim g$.
Furthermore, by the previous observation, we find $h \in \terms$ such that $e^* \satrace{U_n} h$ with $f_n \cdot e^* \lesssim h$.
By Lemma~\ref{lemma:trace-construct-sequential} and the fact that $f_{n-1} \in \sacc$, we find $d \in \terms$ such that $f_{n-1} \cdot e^* \satrace{U_n} d$ with $h \lesssim d$.
By the fact that $f_{n-1} \cdot e^* \lesssim g$ and $f_{n-1} \cdot e^* \satrace{U_n} d$, we find $f \in \terms$ such that $g \satrace{U_n} f$ with $d \lesssim f$.
In total, we have that $e^* \satrace{U} f$ with $f_n \cdot e^* \lesssim h \lesssim d \lesssim f$.
\end{arxivproof}

\subsection{Soundness for the syntactic PA}%
\label{subsection:soundness}

With trace deconstruction and construction lemmas in our toolbox, we are ready to show that the syntactic PA indeed captures series-rational languages.

First, note that $\slang$ can be seen as a function from $\terms$ to $\pom$, like $\sem{-}$.
To establish equality between $\slang$ and $\sem{-}$, we first show that $\slang$ enjoys the same homomorphic equalities as those in the definition of the semantic map, i.e., that $\slang(e)$ can be expressed in terms of $\slang$ applied to subexpressions of $e$.
The proofs of the equalities below follow a similar pattern: for the inclusion from left to right we use trace deconstruction lemmas to obtain traces for the component expressions, while for the inclusion from right to left we use trace construction lemmas to build traces for the composed expressions given the traces of the component expressions.
We treat the case for the empty pomset separately almost everywhere.
\begin{lemma}%
\label{lemma:syntactic-pa-language-homomorphism}
Let $e_1, e_2 \in \terms$, and $a \in \Sigma$.
The following equalities hold:
\ifarxiv%
\begin{align}
\slang(0) &= \emptyset \label{equality:soundness-zero} \\
\slang(1) &= \{ 1 \} \label{equality:soundness-one} \\
\slang(a) &= \{ a \} \label{equality:soundness-primitive} \\
\slang(e_1 + e_2) &= \slang(e_1) \cup \slang(e_2) \label{equality:soundness-plus} \\
\slang(e_1 \cdot e_2) &= \slang(e_1) \cdot \slang(e_2) \label{equality:soundness-sequential} \\
\slang(e_1 \parallel e_2) &= \slang(e_1) \parallel \slang(e_2) \label{equality:soundness-parallel} \\
\slang(e_1^*) &= {\slang(e_1)}^* \label{equality:soundness-star}
\end{align}
\else%
\begin{mathpar}
\slang(0) = \emptyset  \and
\slang(1) = \{ 1 \}    \and
\slang(a) = \{ a \}    \and
\slang(e_1 + e_2) = \slang(e_1) \cup \slang(e_2) \\
\slang(e_1 \cdot e_2) = \slang(e_1) \cdot \slang(e_2) \and
\slang(e_1 \parallel e_2) = \slang(e_1) \parallel \slang(e_2)\and
\slang(e_1^*) = {\slang(e_1)}^*
\end{mathpar}
\fi%
\end{lemma}
\begin{arxivproof}
For~\eqref{equality:soundness-zero}, suppose that $U \in \slang(0)$.
Then there exists an $f \in \sacc$ such that $0 \sartrace{U} f$.
Since $0 \not\in \sacc$, we know that $f \neq 0$, and so $0 \satrace{U} f$ must hold.
But then, by Lemma~\ref{lemma:trace-deconstruct-base} we know that $f = 0$, which is a contradiction.
We conclude that $\slang(0) = \emptyset$.

For~\eqref{equality:soundness-one}, suppose that $U \in \slang(1)$.
Then there exists an $f \in \sacc$ such that $1 \sartrace{U} f$.
Then either $f = 1$ and $U = 1$, or $U \in \pom^+$ and $1 \satrace{U} f$.
The former case is possible (since $1 \in \sacc$), but in the latter case we find that $f = 0$ by Lemma~\ref{lemma:trace-deconstruct-base}, and so $f \not\in \sacc$ --- a contradiction.
In conclusion, we find that $U = 1$, and thus $U \in \{ 1 \}$.
The other inclusion is easy: simply observe that $1 \sartrace{1} 1$ by definition of $\sartracerel$, implying that $1 \in \slang(1)$.

For~\eqref{equality:soundness-primitive}, suppose that $U \in \slang(a)$.
Then there exists an $f \in \sacc$ such that $a \sartrace{U} f$.
Then either $f = a$ and $U = 1$, or $U \in \pom^+$ and $a \satrace{U} f$.
We can rule out the former case, as it implies that $f = a \in \sacc$ --- a contradiction.
In the latter case, we find by Lemma~\ref{lemma:trace-deconstruct-base} that either $U = a$ and $f = 1$, or $f = 0$.
Since the latter again contradicts that $f \in \sacc$, we find that $U = a$; we thus conclude that $U \in \{ a \}$.
The other inclusion is easy: simply observe that $a \satrace{a} 1$ by definition of $\satracerel$, and thus $a \sartrace{a} 1$, implying that $a \in \slang(a)$.

For~\eqref{equality:soundness-plus}, suppose that $U \in \slang(e_1 + e_2)$.
Then there exists an $f \in \sacc$ such that $e_1 + e_2 \sartrace{U} f$.
If $f = e_1 + e_2$ and $U = 1$, then either $e_1 \in \sacc$ or $e_2 \in \sacc$.
In either case, it follows quite easily that $U \in \slang(e_1) \cup  \slang(e_2)$.
Otherwise, $e_1 + e_2 \satrace{U} f$.
But then, by Lemma~\ref{lemma:trace-decompose-plus}, we know that $f = f_1 + f_2$, such that $e_1 \satrace{U} f_1$ and $e_2 \satrace{U} f_2$.
Since $f \in \sacc$, either $f_1 \in \sacc$ or $f_2 \in \sacc$, and therefore $U \in \slang(e_1)$ or $U \in \slang(e_2)$; either way, $U \in \slang(e_1) \cup \slang(e_2)$.

To prove the other inclusion, suppose that $U \in \slang(e_1) \cup \slang(e_2)$.
If $U \in \slang(e_1)$, then there exists an $f_1 \in \sacc$ such that $e_1 \sartrace{U} f_1$.
If $U = 1$ and $f_1 = e_1$, then $e_1 \in \sacc$ and therefore $e_1 + e_2 \in \sacc$; it follows that $U \in \slang(e_1 + e_2)$.
Otherwise, $e_1 \satrace{U} f_1$, and so by Lemma~\ref{lemma:subcongruent-states-simulation} and the fact that $e_1 \lesssim e_1 + e_2$ we find $f \in \terms$ such that $e_1 + e_2 \satrace{U} f$ with $f_1 \lesssim f_2$.
By Lemma~\ref{lemma:subcongruent-states-simulation}, $f \in \sacc$, and so it follows that $U \in \slang(e_1 + e_2)$.
The case where $U \in \slang(e_2)$ is similar.

For~\eqref{equality:soundness-sequential}, suppose that $U \in \slang(e_1 \cdot e_2)$.
Then there exists an $f \in \sacc$ such that $e_1 \cdot e_2 \sartrace{U} f$.
If $U = 1$ and $f = e_1 \cdot e_2$, then $e_1 \cdot e_2 \in \sacc$, and thus $e_1, e_2 \in \sacc$.
Therefore $1 \in \slang(e_1)$ and $1 \in \slang(e_2)$, implying that $1 = 1 \cdot 1 \in \slang(e_1) \cdot \slang(e_2)$.
Otherwise, $e_1 \cdot e_2 \satrace{U} f$.
Then, by Lemma~\ref{lemma:trace-deconstruct-sequential}, there exists an $f' \in \terms$ and a finite set $I$, as well as $I$-indexed families ${(f_i')}_{i \in I}$ over $\sacc$ and ${(f_i)}_{i \in I}$ over $\terms$ and ${(U_i')}_{i \in I}$ over $\pom$ and ${(U_i)}_{i \in I}$ over $\terms$, such that:
\begin{itemize}
    \item $f \simeq f' \cdot e_2 + \sum_{i \in I} f_i$, $e_1 \satrace{U} f'$
    \item for all $i \in I$, $e_1 \satrace{U_i'} f_i'$, $e_2 \satrace{U_i} f_i$ and $U = U_i' \cdot U_i$.
\end{itemize}
Since $f \in \sacc$, also $f' \cdot e_2 + \sum_{i \in I} f_i \in \sacc$ by Lemma~\ref{lemma:congruent-states-halt}.
This means that either $f' \cdot e_2 \in \sacc$, or $f_i \in \sacc$ for some $i \in I$.
In the former case, $f', e_2 \in \sacc$, and therefore $U \in \slang(e_1)$ and $1 \in \slang(e_2)$, thus $U = U \cdot 1 \in \slang(e_1) \cdot \slang(e_2)$.
In the latter case, $U_i' \in \slang(e_1)$ and $U_i \in \slang(e_2)$, thus $U = U_i' \cdot U_i \in \slang(e_1) \cdot \slang(e_2)$.

To prove the other inclusion, suppose that $U \in \slang(e_1) \cdot \slang(e_2)$.
Then $U = U_1 \cdot U_2$, with $U_1 \in \slang(e_1)$ and $U_2 \in \slang(e_2)$.
There are four cases to consider.
\begin{itemize}
    \item If $U_1 = 1$ and $U_2 = 1$, then $e_1, e_2 \in \sacc$, and therefore $e_1 \cdot e_2 \in \sacc$.
    Thus $1 \in \slang(e_1 \cdot e_2)$.
    \item If $U_1 = 1$ and $U_2 \neq 1$, then $e_1 \in \sacc$ and $e_2 \satrace{U_2} f_2$ for some $f_2 \in \sacc$.
    Then, by Lemma~\ref{lemma:trace-construct-sequential}, $e_1 \cdot e_2 \satrace{U_2} f$ for some $f \in \terms$ with $f_2 \lesssim f$.
    By Lemma~\ref{lemma:subcongruent-states-simulation}, $f \in \sacc$, and thus $U = U_2 \in \slang(e_1 \cdot e_2)$.
    \item If $U_1 \neq 1$ and $U_2 = 1$, then $e_1 \satrace{U_1} f_1$ for some $f_1 \in \sacc$, and $e_2 \in \sacc$.
    Then, by Lemma~\ref{lemma:trace-construct-sequential}, $e_1 \cdot e_2 \satrace{U_1} f$ for some $f \in \terms$ with $f_1 \cdot e_2 \lesssim f$.
    Since $f_1, e_2 \in \sacc$, also $f_1 \cdot e_2 \in \sacc$, and thus by Lemma~\ref{lemma:subcongruent-states-simulation}, $f \in \sacc$.
    Therefore $U = U_1 \in \slang(e_1 \cdot e_2)$.
    \item If $U_1 \neq 1$ and $U_2 \neq 1$, then $e_1 \satrace{U_1} f_1$ and $e_2 \satrace{U_2} f_2$ for some $f_1, f_2 \in \sacc$.
    Then, by Lemma~\ref{lemma:trace-construct-sequential}, we find that $e_1 \cdot e_2 \satrace{U} f$ with $f_2 \lesssim f$.
    By Lemma~\ref{lemma:subcongruent-states-simulation}, $f \in \sacc$, and thus $U \in \slang(e_1 \cdot e_2)$.
\end{itemize}

For~\eqref{equality:soundness-parallel}, suppose that $U \in \slang(e_1 \parallel e_2)$.
Then there exists an $f \in \sacc$ such that $e_1 \parallel e_2 \sartrace{U} f$.
If $f = e_1 \parallel e_2$ and $U = 1$, then $e_1, e_2 \in \sacc$, and thus $1 \in \slang(e_1)$ and $1 \in \slang(e_2)$.
But then $1 = 1 \parallel 1 \in \slang(e_1) \parallel \slang(e_2)$.
Otherwise, $e_1 \parallel e_2 \satrace{U} f$ for some $f \in \sacc$.
By Lemma~\ref{lemma:trace-deconstruct-parallel}, we obtain $f_1, f_2, f_3 \in \sacc$ such that
\begin{itemize}
    \item $f \simeq f_1 + f_2 + f_3$
    \item either $f_1 = 0$, or $e_2 \in \sacc$ and $e_1 \satrace{U} f_1$
    \item either $f_2 = 0$, or $e_1 \in \sacc$ and $e_2 \satrace{U} f_2$
    \item either $f_3 = 0$, or $f_3 = 1$ and there exist $e_1', e_2' \in \terms$ and $f_1', f_2' \in \sacc$ and $U_1, U_2 \in \pom^+$ such that $e_1 \simeq e_1'$ and $e_2 \simeq e_2'$ and $U = U_1 \parallel U_2$ and $e_1' \satrace{U_1} f_1'$ and $e_2' \satrace{U_2} f_2'$.
\end{itemize}
By Lemma~\ref{lemma:subcongruent-states-simulation}, $f_1 + f_2 + f_3 \in \sacc$, and thus $f_1 \in \sacc$ or $f_2 \in \sacc$ or $f_3 \in \sacc$.
In the first case, $U \in \slang(e_1)$ and $1 \in \slang(e_2)$, and therefore $U = U \parallel 1 \in \slang(e_1) \parallel \slang(e_2)$.
In the second case, we similarly find that $U \in \slang(e_1) \parallel \slang(e_2)$.
In the last case, we find by Lemma~\ref{lemma:congruent-states-progress} and Lemma~\ref{lemma:congruent-states-halt} $f_1'', f_2'' \in \sacc$ such that $e_1 \satrace{U} f_1''$ and $e_2 \satrace{U} f_2''$, and thus we have that $U = U_1 \parallel U_2$ such that $U_1 \in \slang(e_1)$ and $U_2 \in \slang(e_2)$, therefore $U \in \slang(e_1) \parallel \slang(e_2)$.

To prove the other inclusion, suppose that $U \in \slang(e_1) \parallel \slang(e_2)$.
Then $U = U_1 \parallel U_2$.
There are four cases to consider.
\begin{itemize}
    \item If $U_1 = 1$ and $U_2 = 1$, then $e_1, e_2 \in \sacc$, thus $e_1 \parallel e_2 \in \sacc$.
    But then $U = 1 \in \slang(e_1 \parallel e_2)$.
    \item If $U_1 = 1$ and $U_2 \neq 1$, then $e_2 \satrace{U_2} f_2$ for some $f_2 \in \sacc$.
    By Lemma~\ref{lemma:trace-construct-parallel}, we find an $f \in \terms$ such that $e_1 \parallel e_2 \satrace{U} f$ with $f_2 \lesssim f$.
    But then, by Lemma~\ref{lemma:subcongruent-states-simulation}, $f \in \sacc$, and thus $U \in \slang(e_1 \parallel e_2)$.
    \item If $U_1 \neq 1$ and $U_2 = 1$, then we find that $U \in \slang(e_1 \parallel e_2)$ by an argument similar to the above case.
    \item If $U_1 \neq 1$ and $U_2 \neq 1$, then $e_1 \satrace{U_1} f_1$ and $e_2 \satrace{U_2} f_2$ for some $f_1, f_2 \in \sacc$.
    By Lemma~\ref{lemma:trace-construct-parallel}, we find an $f \in \terms$ such that $e_1 \parallel e_2 \satrace{U} f$ with $1 \lesssim f$.
    But then, by Lemma~\ref{lemma:subcongruent-states-simulation}, $f \in \sacc$, and thus $U \in \slang(e_1 \parallel e_2)$.
\end{itemize}

For~\eqref{equality:soundness-star}, suppose that $U \in \slang(e_1^*)$.
Then there exists an $f \in \sacc$ such that $e_1^* \sartrace{U} f$.
If $f = e_1^*$ and $U = 1$, then $1 \in {\slang(e_1)}^*$.
Otherwise, $e_1 \satrace{U} f$.
By Lemma~\ref{lemma:trace-deconstruct-star}, we find a finite set $I$ and an $I$-indexed family of finite sets ${(J_i)}_{i \in I}$ as well as $I$-indexed families ${(f_i)}_{i \in I}$ over $\terms$ and ${(U_i)}_{i \in I}$ over $\pom^+$, and for all $i \in I$ also $J_i$-indexed families ${(f_{i,j})}_{j \in J_i}$ over $\terms$ and ${(U_{i,j})}_{j \in J_i}$ over $\pom^+$, such that $f \simeq \sum_{i \in I} f_i \cdot e^*$, and for all $i \in I$
\begin{itemize}
    \item $e \satrace{U_i} f_i$
    \item for all $j \in J_i$ we have that $e \satrace{U_{i,j}} f_{i,j}$
    \item $U = U_i' \cdot U_i$, where $U_i'$ is some concatenation of all $U_{i,j}$ for all $j \in J_i$.
\end{itemize}
By Lemma~\ref{lemma:subcongruent-states-simulation}, we know that $\sum_{i \in I} f_i \cdot e^* \in \sacc$, and thus $f_i \cdot e^* \in \sacc$ for some $i \in I$, meaning in particular that $f_i \in \sacc$ for this $i \in I$.
Now $U_{i,j} \in \slang(e)$ for all $j \in J$, and thus $U_i' \in {\slang(e)}^*$.
Since also $U_i \in \slang(e)$, we find that $U = U_i' \cdot U_i \in {\slang(e)}^* \cdot \slang(e) \subseteq {\slang(e)}^*$.

To prove the other inclusion, suppose that $U \in {\slang(e)}^*$.
Then $U \in {\slang(e)}^n$, i.e., $U = U_1 \cdot U_2 \cdots U_n$, such that for all $1 \leq i \leq n$ we have $U_i \in \slang(e)$.
Assume (without loss of generality) that $U_i \neq 1$ for all $1 \leq i \leq n$.
If $n = 0$, then $U = 1$ and we find that $U \in \slang(e^*)$ immediately.
Otherwise, there exist $f_1, f_2, \dots, f_n \in \sacc$ such that $e \satrace{U_i} f_i$.
By Lemma~\ref{lemma:trace-construct-star} we find $f \in \terms$ such that $e^* \satrace{U} f$ with $f_n \lesssim f$.
But then $f \in \sacc$ by Lemma~\ref{lemma:subcongruent-states-simulation}, and thus $U \in \slang(e^*)$.
\end{arxivproof}

It is now easy to establish that the Brzozowski construction for the syntactic PA is sound with respect to the denotational semantics of series-rational expressions.
\begin{theorem}%
\label{theorem:syntactic-pa-soundness}
Let $e \in \terms$.
Then $\slang(e) = \sem{e}$.
\end{theorem}
\begin{proof}
The proof proceeds by induction on $e$.
In the base, $e = 0$, $e = 1$ or $e = a$ for some $a \in \Sigma$.
In all cases, $\slang(e) = \sem{e}$ by Lemma~\ref{lemma:syntactic-pa-language-homomorphism}.
For the inductive step, there are four cases to consider: either $e = e_1 + e_2$, $e = e_1 \cdot e_2$, $e = e_1 \parallel e_2$ or $e = e_1^*$.
In all cases, the claim follows from the induction hypothesis and the definition of $\sem{-}$, combined with Lemma~\ref{lemma:syntactic-pa-language-homomorphism}.
\end{proof}

\subsection{Bounding the syntactic PA}%
\label{subsection:bounding}

Ideally, we would like to obtain a single PA with finitely many states that recognizes $\sem{e}$ for a given $e \in \terms$.
Unfortunately, the syntactic PA is not bounded, and thus Theorem~\ref{theorem:bounded-finite-subpa} does not apply.
For instance, the requirement that $\rho_\Sigma(e)$ be finite for $e \in \terms$ fails; consider the family of distinct terms ${(e_n)}_{n \in \naturals}$ defined by $e_0 = 1 \cdot a^*$ and $e_{n+1} = 0 \cdot a^* + e_n$ for $n \in \naturals$; it is not hard to show that $e_n \in \rho_\Sigma(a^*)$ for $n \in \naturals$, and thus conclude that $\rho_\Sigma(a^*)$ is infinite.
We remedy this problem by quotienting the state space of the syntactic PA by congruence.

In what follows, we write $[e]$ for the congruence class of $e \in \terms$ modulo $\simeq$, i.e., the set of all $e' \in \terms$ such that $e \simeq e'$.
We furthermore write $\qterms$ for the set of all congruence classes of expressions in $\terms$.
We now leverage Lemma~\ref{lemma:congruent-states-derive} to define a transition structure on $\qterms$.
\ifnotarxiv%

To save space, this section only summarizes the main stepping stones towards finding a finite PA for an expression; for a full proof, we refer to the full version of this paper~\cite{kappe-brunet-luttik-silva-zanasi-2017-arxiv}.
\fi%
\begin{definition}
We define $\qssderiv: \qterms \times \Sigma \to \qterms$ and $\qpsderiv: \qterms \times \binom{\qterms}{2} \to \qterms$ as
\begin{mathpar}
\qssderiv([e], a) = [\ssderiv(e, a)]
\and%
\qpsderiv([e], \mset{[f], [g]}) =
\begin{cases}
[0] & f \simeq 0\ \mbox{or}\ g \simeq 0 \\
[\psderiv(e, \mset{g, h})] & \mbox{otherwise}
\end{cases}
\end{mathpar}
Furthermore, the set $\qsacc$ is defined to be $\{ [e] : e \in \sacc \}$.
The \emph{quotiented syntactic PA} is the PA $A_\simeq = \angl{\qterms, \qssderiv, \qpsderiv, \qsacc}$.
\end{definition}

Note that, by virtue of Lemma~\ref{lemma:congruent-states-derive} and Lemma~\ref{lemma:congruent-states-halt}, we have that $\qssderiv$ and $\qpsderiv$, as well as $\qsacc$, are well-defined.
As before, we abbreviate subscripts, for example by writing $\qsatracerel$ rather than $\tracerel_{A_\simeq}$, and $\qslang$ rather than $L_{A_\simeq}$.
Of course, we also want the quotiented syntactic PA to accept the same languages as the syntactic PA\@.
\ifarxiv%
To that end, we show that the trace relations of the syntactic PA and the quotiented syntactic PA correspond.

\begin{restatable}{lemma}{syntracevsquotsyntrace}%
\label{lemma:syntactic-trace-vs-quotiented-syntactic-trace}
Let $e, f \in \terms$ and $U \in \pom^+$.
If $e \satrace{U} f$, then $[e] \qsatrace{U} [f]$.
If $[e] \qsatrace{U} [f]$, then there exists an $f \in \terms$ with $f \simeq f'$ and $e \satrace{U} f'$.
\end{restatable}
\begin{proof}
Refer to Appendix~\ref{appendix:proofs}.
\end{proof}
\else%
This turns out to be the case.
\fi%

\begin{theorem}%
\label{theorem:quotiented-syntactic-pa-vs-syntactic-pa}
Let $e \in \terms$.
Then $\slang(e) = \qslang([e])$.
\end{theorem}
\begin{arxivproof}
By definition, $e \in F_\Sigma$ if and only if $[e] \in F_\simeq$.
It then follows that $1 \in \slang(e)$ if and only if $1 \in \qslang([e])$.
The remaining cases are covered by Lemma~\ref{lemma:syntactic-trace-vs-quotiented-syntactic-trace}.
\end{arxivproof}

\ifarxiv%
\begin{corollary}%
\label{corollary:quotiented-syntactic-pa-single-empty-state}
The state $[0]$ is the only state in the quotiented syntactic PA with an empty language.
\end{corollary}
\begin{proof}
Let $[e] \in \qterms$ be a state such that $\qslang([e]) = \emptyset$.
Then, by Theorem~\ref{theorem:quotiented-syntactic-pa-vs-syntactic-pa}, we know that $\slang(e) = \emptyset$, and by Theorem~\ref{theorem:syntactic-pa-soundness}, $\sem{e} = \emptyset$.
But then $e \simeq 0$ by Lemma~\ref{lemma:congruence-sound}, and thus $[e] = [0]$.
\end{proof}

We now show that the quotiented syntactic PA is bounded.
First, we need the following.

\begin{definition}%
\label{definition:parallel-depth}
Let $e \in \terms$.
The \emph{parallel depth} of $e$, denoted $\depth{e}$, is $0$ when $e \simeq 0$, and otherwise:
\begin{mathpar}
\depth{1} = 0 \and
\depth{a} = 1 \and
\depth{e_1 + e_2} = \max(\depth{e_1}, \depth{e_2})\\
\depth{e_1 \cdot e_2} = \max(\depth{e_1}, \depth{e_2})\and
\depth{e_1 \parallel e_2} = \max(\depth{e_1}, \depth{e_2}) + 1\and
\depth{e_1^*} = \depth{e_1}
\end{mathpar}
\end{definition}

It is easy to show that the parallel depth of an expression is also well-defined on the congruence classes of $\simeq$.
This allows us to define a fork hierarchy on $\qterms$.
\begin{restatable}{lemma}{depthvscongruence}%
\label{lemma:depth-vs-congruence}
Let $e, e' \in \terms$ be such that $e \simeq e'$.
Then $\depth{e} = \depth{e'}$.
\end{restatable}
\begin{proof}
Refer to Appendix~\ref{appendix:proofs}.
\end{proof}

\begin{definition}
The relation $\prec\ \subseteq \qterms \times \qterms$ is such that $[e] \prec [f]$ if and only if $\depth{e} < \depth{f}$.
\end{definition}

\begin{lemma}%
\label{lemma:quotiented-syntactic-pa-fork-acyclic}
The quotiented syntactic PA is fork-acyclic, with fork hierarchy $\prec$.
\end{lemma}
\begin{proof}
We need to show that $\prec$ satisfies the conditions of Definition~\ref{definition:fork-acyclic}.

For the first rule, we begin by showing that if $g, h, e \in \terms$ such that $e, g, h \not\simeq 0$ and $\psderiv(e, \mset{g, h}) \not\simeq 0$, then $\depth{g} < \depth{e}$.
The proof proceeds by induction on $e$.
In the base, $e = 1$ or $e = a$; in both cases, $\psderiv(e, \mset{g, h}) = 0$ and so the claim holds vacuously.
For the inductive step, there are four cases to consider.
\begin{itemize}
    \item If $e = e_1 + e_2$, then $\psderiv(e, \mset{g, h}) = \psderiv(e_1, \mset{g, h}) + \psderiv(e_2, \mset{g, h}) \not\simeq 0$, thus $\psderiv(e_1, \mset{g, h}) \not\simeq 0$ or $\psderiv(e_2, \mset{g, h}) \not\simeq 0$.
    Therefore, by induction $\depth{g} < \depth{e_1}$ or $\depth{g} < \depth{e_2}$, and thus $\depth{g} < \max(\depth{e_1}, \depth{e_2}) = \depth{e}$.
    \item If $e = e_1 \cdot e_2$, then $\psderiv(e, \mset{g, h}) = \psderiv(e_1, \mset{g, h}) \cdot e_2 + e_1 \star \psderiv(e_2, \mset{g, h}) \not\simeq 0$, thus $\psderiv(e_1, \mset{g, h}) \not\simeq 0$ or $\psderiv(e_2, \mset{g, h}) \not\simeq 0$.
    Therefore, by induction $\depth{g} < \depth{e_1}$ or $\depth{g} < \depth{e_2}$, and thus $\depth{g} < \max(\depth{e_1}, \depth{e_2}) = \depth{e}$.
    \item If $e = e_1 \parallel e_2$, then $\psderiv(e, \mset{g, h}) = [\mset{g, h} = \mset{e_1, e_2}] + e_2 \star \psderiv(e_1, \mset{g, h}) + e_1 \star \psderiv(e_2, \mset{g, h}) \not\simeq 0$.
    If $[\mset{g, h} = \mset{e_1, e_2}] \not\simeq 0$, we know that $\mset{g, h} = \mset{e_1, e_2}$.
    Assume without loss of generality that $e_1 = g$; then $\depth{g} = \depth{e_1} < \max(\depth{e_1}, \depth{e_2}) + 1 = \depth{e}$.
    If $\psderiv(e_1, \mset{g, h}) \not\simeq 0$, then by induction $\depth{g} < \depth{e_1} < \max(\depth{e_1}, \depth{e_2}) + 1$.
    The case where $\psderiv(e_2, \mset{g, h}) \not\simeq 0$ is similar.
    \item If $e = e_1^*$, then $\psderiv(e, \mset{g, h}) = \psderiv(e_1, \mset{g, h}) \cdot e_1^* \not\simeq 0$ and therefore $\psderiv(e_1, \mset{g, h}) \not\simeq 0$.
    But then $\depth{g} < \depth{e_1} = \depth{e}$ by induction.
\end{itemize}
To fully validate the first rule, suppose that $\mset{[g], [h]} \in \qssupport([e])$.
Then $e, g, h \not\simeq 0$ and $\qpsderiv([e], \mset{[g], [h]}) \neq [0]$ and thus $\psderiv(e, \mset{g, h}) \not\simeq 0$.
But then we find that $\depth{g} < \depth{e}$ by the above, and therefore $[g] \prec [e]$.

For the second rule, we first show that if $a \in \Sigma$ and $e \in \terms$, then $\depth{\ssderiv(e, a)} \leq \depth{e}$.
If $e \simeq 0$, then $\depth{\ssderiv(e, a)} = \depth{\ssderiv(0, a)} = \depth{0} = 0 \leq \depth{0} = \depth{e}$.
The proof of the remaining cases proceeds by induction on $e$.
In the base, $e = 1$ or $e = a$.
In both cases, $\depth{\ssderiv(e, a)} = 0 \leq \depth{e}$.

For the inductive step, there are four cases to consider.
\begin{itemize}
    \item If $e = e_1 + e_2$, then we can derive
    \[\depth{\ssderiv(e, a)} = \depth{\ssderiv(e_1, a) + \ssderiv(e_2, a)} = \max(\depth{\ssderiv(e_1, a)}, \depth{\ssderiv(e_2, a)}) \leq \max(\depth{e_1}, \depth{e_2}) = \depth{e}\]
    \item If $e = e_1 \cdot e_2$, then we can derive
    \begin{align*}
    \depth{\ssderiv(e, a)}
        &= \depth{\ssderiv(e_1, a) \cdot e_2 + e_1 \star \ssderiv(e_2, a)} \\
        &\leq \max(\depth{\ssderiv(e_1, a)}, \depth{e_2}, \depth{\ssderiv(e_2, a)}) \\
        &\leq \max(\depth{e_1}, \depth{e_2}) = \depth{e}
    \end{align*}
    \item If $e = e_1 \parallel e_2$, then we can derive
    \begin{align*}
    \depth{\ssderiv(e, a)}
        &= \depth{e_2 \star \ssderiv(e_1, a) + e_1 \star \ssderiv(e_2, a)} \\
        &\leq \max(\depth{\ssderiv(e_1, a)}, \depth{\ssderiv(e_2, a)}) \\
        &\leq \max(\depth{e_1}, \depth{e_2}) = \depth{e_1 \parallel e_2}
    \end{align*}
    \item If $e = e_1^*$, then we can derive
    \[\depth{\ssderiv(e, a)} = \depth{\ssderiv(e_1, a) \cdot e_1^*} \leq \max(\depth{\ssderiv(e_1, a)}, \depth{e_1^*}) \leq \depth{e}\]
\end{itemize}

Thus, suppose that $[f] \prec \qssderiv([e], a) = [\ssderiv(e, a)]$.
Then $[f] \prec [e]$, since we can derive that $\depth{f} < \depth{\ssderiv(e, a)} \leq \depth{e}$.

For the third rule, one also first shows that for $e \in \terms$ and $\phi \in \binom{\terms}{2}$ we have that $\depth{\psderiv(e, \phi)} \leq \depth{e}$.
One can then extend this to show that if $[f] \prec \qpsderiv([e], \phi)$, also $[f] \prec [e]$ using the same technique as above.
\end{proof}

Finally, we investigate the reach and support of a state in the quotiented syntactic PA, using the deconstruction lemmas showed in Section~\ref{subsection:trace-deconstruction}, as well as Lemma~\ref{lemma:syntactic-trace-vs-quotiented-syntactic-trace}.
\begin{lemma}%
\label{lemma:quotiented-syntactic-pa-finite-reach-support}
Let $e \in \terms$.
Then $\rho_\simeq([e])$ and $\qssupport([e])$ are finite.
\end{lemma}
\begin{proof}
We first claim that $R_\Sigma(e)$ is finite up to $\simeq$, i.e., that the set $\{ [f] : f \in R_\Sigma(e) \}$ is finite.
The proof of this claim proceeds by induction on $e$.
In the base, there are three cases to consider.
\begin{itemize}
    \item If $e = 0$ and $e \satrace{U} f$, then $f = 0$ by Lemma~\ref{lemma:trace-deconstruct-base}; it follows that $R_\Sigma(e)$ is finite.
    \item If $e = 1$ and $e \satrace{U} f$, then $f = 0$ by Lemma~\ref{lemma:trace-deconstruct-base}; it follows that $R_\Sigma(e)$ is finite.
    \item If $e = a$ for some $a \in \Sigma$, then $f = 0$ or $f = 1$ by Lemma~\ref{lemma:trace-deconstruct-base}.
    It follows that $R_\Sigma(e)$ is finite.
\end{itemize}
In all cases, we find that $R_\Sigma(e)$ is finite, and therefore finite up to $\simeq$.

For the inductive step, there are four cases to consider.
In each case our strategy is to write an arbitrary element of $R_\Sigma(e)$ into a sum of a number of terms.
If we can then show that there are only finitely many choices for each term up to $\simeq$, it follows that there are only finitely many such sums up to $\simeq$.
Note that the number of terms does not matter, since repeated terms in a sum do not make a difference with respect to $\simeq$ by virtue of the idempotence rule.
\begin{itemize}
    \item If $e = e_1 + e_2$ and $f \in R_\Sigma(e)$, then by Lemma~\ref{lemma:trace-decompose-plus} there exist $f_1 \in R_\Sigma(e_1)$ and $f_2 \in R_\Sigma(e_2)$ such that $f \simeq f_1 + f_2$.
    Since $R_\Sigma(e_1)$ and $R_\Sigma(e_2)$ are finite up to $\simeq$ by induction, it follows that $R_\Sigma(e)$ is finite up to $\simeq$.
    \item If $e = e_1 \cdot e_2$ and $f \in R_\Sigma(e)$, then by Lemma~\ref{lemma:trace-deconstruct-sequential} there exists an $f' \in R_\Sigma(e_1) \cup \{ e_1 \}$ and a finite set $I$ and an $I$-indexed set of terms ${(f_i)}_{i \in I} \in R_\Sigma(e_2)$ such that $f \simeq f' \cdot e_2 + \sum_{i \in I} f_i$.
    Since $R_\Sigma(e_1) \cup \{ e_1 \}$ and $R_\Sigma(e_2)$ are finite up to $\simeq$ by induction, it follows that $R_\Sigma(e)$ is finite up to $\simeq$.
    \item If $e = e_1 \parallel e_2$ and $f \in R_\Sigma(e)$, then by Lemma~\ref{lemma:trace-deconstruct-parallel} there exist $f_1 \in R_\Sigma(e_1) \cup \{ 0 \}$, $f_2 \in R_\Sigma(e_2) \cup \{ 0 \}$ and $f_3 \in \{ 0, 1 \}$ such that $f \simeq f_1 + f_2 + f_3$.
    Since $R_\Sigma(e_1)$ and $R_\Sigma(e_2)$ are finite up to $\simeq$ by induction, it follows that $R_\Sigma(e)$ is finite up to $\simeq$.
    \item If $e = e_1^*$ and $f \in R_\Sigma(e)$, then by Lemma~\ref{lemma:trace-deconstruct-star} there exists a finite set $I$ and an $I$-indexed family ${(f_i)}_{i \in I}$ over $R_\Sigma(e_1)$ such that $f \simeq \sum_{i \in I} f_i \cdot e^*$.
    Since $R_\Sigma(e_1)$ is finite up to $\simeq$ by induction, it follows that $R_\Sigma(e)$ is finite up to $\simeq$.
\end{itemize}

Since $\rho_\simeq(e) \cup \{ [0] \} = R_\simeq([e]) \cup \{ [e], [0] \}$ by Lemma~\ref{lemma:reach-vs-reachable} and Corollary~\ref{corollary:quotiented-syntactic-pa-single-empty-state}, it suffices to show that $R_\simeq([e])$ is finite.
This is an immediate consequence of Lemma~\ref{lemma:syntactic-trace-vs-quotiented-syntactic-trace} and the above.

We now treat the claim about $\qssupport([e])$.
First, we claim that there are finitely many $\phi \in \binom{\terms}{2}$ up to $\simeq$ such that $\psderiv(e, \phi) \not\simeq 0$.
 This is shown by induction on $e$.
In the base, $e = 0$, $e = 1$ or $e = a$; in all cases, we find that $\psderiv(e, \phi) = 0$ for $\phi \in \binom{\terms}{2}$.
For the inductive step, there are four cases; in each case, we apply Lemma~\ref{lemma:syntactic-pa-language-homomorphism}.
\begin{itemize}
    \item Suppose $e = e_1 + e_2$.
    By induction, there are finitely many $\phi$ up to $\simeq$ such that $\psderiv(e_1, \phi) \not\simeq 0$ or $\psderiv(e_2, \phi) \not\simeq 0$.
    Thus, there are finitely many $\phi \in \binom{\terms}{2}$ up to $\simeq$ such that $\psderiv(e, \phi) = \psderiv(e_1, \phi) + \psderiv(e_2, \phi) \not\simeq 0$.
    \item Suppose $e = e_1 \cdot e_2$.
    By induction, there are finitely many $\phi$ up to $\simeq$ such that $\psderiv(e_1, \phi) \not\simeq 0$ or $\psderiv(e_2, \phi) \not\simeq 0$.
    Thus, there are finitely many $\phi \in \binom{\terms}{2}$ up to $\simeq$ such that $\psderiv(e, \phi) = \psderiv(e_1, \phi) \cdot e_2 + e_1 \star \psderiv(e_2, \phi) \not\simeq 0$.
    \item Suppose $e = e_1 \parallel e_2$.
    By induction, there are finitely many $\phi$ up to $\simeq$ such that $\psderiv(e_1, \phi) \not\simeq 0$ or $\psderiv(e_2, \phi) \not\simeq 0$.
    Furthermore, there is exactly one $\phi$ up to $\simeq$ such that $\phi \simeq \mset{e_1, e_2}$.
    Thus, there are finitely many $\phi \in \binom{\terms}{2}$ up to $\simeq$ such that $\psderiv(e, \phi) = [\phi \simeq \mset{e_1, e_2}] + e_2 \star \psderiv(e_1, \phi) + e_1 \star \psderiv(e_2, \phi) \not\simeq 0$.
    \item Suppose $e = e_1^*$.
    By induction, there are finitely many $\phi$ up to $\simeq$ such that $\psderiv(e_1, \phi) \not\simeq 0$.
    Thus, there are finitely many $\phi \in \binom{\terms}{2}$ up to $\simeq$ such that $\psderiv(e, \phi) = \psderiv(e_1, \phi) \cdot e_1^* \not\simeq 0$.
\end{itemize}
The proof that that $\qssupport([e])$ is finite, i.e., that there are finitely many $\phi \in \binom{\qterms}{2}$ such that $\qpsderiv([e], \phi) \neq [0]$ is now an immediate consequence of the above.
\end{proof}

\begin{theorem}%
\label{theorem:quotiented-syntactic-pa-bounded}
The quotiented syntactic PA is bounded.
\end{theorem}
\begin{proof}
By Lemma~\ref{lemma:quotiented-syntactic-pa-fork-acyclic}, we know that $A_\simeq$ is fork-acyclic.
Furthermore, $\prec$ is well-founded by construction: every term has finite depth.
By Lemma~\ref{lemma:quotiented-syntactic-pa-finite-reach-support}, we know that $\rho_\simeq([e])$ and $\qssupport([e])$ are finite for all $[e] \in \qterms$.
All requirements are now validated.
\end{proof}
\else%
Furthermore, the quotiented syntactic PA is sufficiently restricted to show the following:
\begin{theorem}%
\label{theorem:quotiented-syntactic-pa-bounded-fork-acyclic}
The quotiented syntactic PA is fork-acyclic and bounded.
\end{theorem}
\fi%

The desired result then follows from the above, Lemma~\ref{lemma:subpa-preserves-language} and Theorem~\ref{theorem:bounded-finite-subpa}.
\begin{corollary}
Let $e \in \terms$.
There exists a finite PA $A_e$ that accepts $\sem{e}$.
\end{corollary}
\begin{arxivproof}
Consider the quotiented syntactic PA $A_\simeq$.
By Theorem~\ref{theorem:quotiented-syntactic-pa-vs-syntactic-pa} and Theorem~\ref{theorem:syntactic-pa-soundness}, we know that $\qslang([e]) = \slang(e) = \sem{e}$.
By Theorem~\ref{theorem:quotiented-syntactic-pa-bounded}, $A_\simeq$ is bounded, and therefore (by Theorem~\ref{theorem:bounded-finite-subpa}) we can construct a finite and closed set $Q_e$ such that $[e] \in Q_e$.
We now choose $A_e = \restr{A_\simeq}{Q_e}$ and $q = [e]$ to find that $L_{A_e}(q) = \qslang([e]) = \sem{e}$ by Lemma~\ref{lemma:subpa-preserves-language}.
\end{arxivproof}

\section{Automata to expressions}%
\label{section:automata-to-expressions}

To associate with every state $q$ in a bounded PA $A = \angl{Q, \delta, \gamma, F}$ a series-rational expression $e_q$ such that $\sem{e_q}=L_A(q)$, we modify the procedure for associating a rational expression with a state in a finite automaton described in~\cite{kozen-1997-a}.
The modification consists of adding parallel terms to the expression associated with $q$ whenever a fork in $q$ contributes to its language, i.e., whenever $\mset{r,s}\in\pi_A(q)$.

In view of the special treatment of $1$ in the semantics of PAs, it is convenient to first define expressions $e_q^+$ with the property that $\sem{e_q^+}=L_A(q)-\{1\}$; then we can define $e_q$ by $e_q=e_q^++[q\in F]$.
The definition of $e_q^+$ proceeds by induction on the well-founded partial order $\prec_A$ associated with a bounded PA\@.
That is, when defining $e_q^+$ we assume the existence of expressions $e_{q'}^+$ for all $q'\in Q$ such that $q'\prec_A q$.

First, however, we shall define auxiliary expressions $e^{Q'}_{qq'}$ for suitable choices of $Q'\subseteq Q$ and of $q,q'\in Q$.
Intuitively, $e^{Q'}_{qq'}$ denotes the pomset language characterizing all paths from $q$ to $q'$ with all intermediate states in $Q'$; $e_q^+$ can then be defined as the summation of all $e^{\rho_A(q)}_{qq'}$ with $q'\in F \cap \rho_A(q)$.

\begin{definition}\label{definition:autexpraux}
Let $Q'$ be a finite subset of $Q$, and assume that for all $r \in Q$ such that $r \prec_A q$ for some $q \in Q'$ there exists a series-rational expression $e_r^+\in\terms$ such that $\sem{e_{r}^+}=L_A(q)-\{1\}$.
For all $Q''\subseteq Q'$ and $q,q'\in Q'$, we define a series-rational expression $e^{Q''}_{qq'}$ by induction on the size of $Q''$, as follows:
\begin{enumerate}
    \item\label{basis} If $Q'' = \emptyset$, then let $\widetilde{\Sigma} = \{ a \in \Sigma : q' = \delta(q, a) \}$, and let $\widetilde{Q} = \{ \phi \in \pi_A(q) : \gamma(q, \phi) = q' \}$.
    We define
    \[e^{Q''}_{qq'} = \sum_{a \in \widetilde{\Sigma}} a + \sum_{\mset{r, s} \in \widetilde{Q}} e_r^+\parallel e_s^+ \enskip.\]
    \item\label{induction} Otherwise, we choose a $q''\in Q''$ and define
    \[e^{Q''}_{qq'} =  e^{Q''-\{q''\}}_{qq'} +e^{Q''-\{q''\}}_{qq''} \cdot {(e^{Q''-\{q''\}}_{q''q''})}^{*} \cdot e^{Q''-\{q''\}}_{q''q'} \enskip.\]
\end{enumerate}
\end{definition}
Note that $e_r^+$ and $e_s^+$, appearing in the first clause of the definition of $e^{Q''}_{qq'}$, exist by assumption, for by fork-acyclicity we have that $r,s\prec_A q\in Q'$.

\begin{theorem}
Let $Q'$ be a finite subset of $Q$ and assume that for all $r \in Q$ such that $r \prec_A q$ for some $q \in Q'$ there exists a series-rational expression $e_r^+\in\terms$ with $\sem{e_r^+}=L_A(q)-\{1\}$.
For all $q, q'\in Q'$, for all $Q''\subseteq Q'$, and for all $U\in\pom^+$, we have that $q\atrace{U}_A q'$ according to some path that only visits states in $Q''$ if, and only if, $U\in\sem{e^{Q''}_{qq'}}$.
\end{theorem}
\begin{arxivproof}
We prove the implication from left to right with induction on the size of $Q''$.

Suppose that $q \atrace{U}_A q'$ according to some path through $A$ that only visits states in $Q''$.
Then there exist $n \geq 1$, $q_0,q_{n} \in Q$, $q_1, \dots, q_{n-1} \in Q''$ and $U_1, \dots, U_{n} \in \pom^+$ such that $U = U_1 \cdots U_{n}$ and
\[q = q_0 \atrace{U_1} q_1 \atrace{U_2} \cdots \atrace{U_{n-1}} q_{n-1} \atrace{U_n} q_{n} = q'\]
Furthermore, for all $1 \leq i \leq n$, either $U_i = a$ for some $a \in \Sigma$ and $q_i = \delta(q_{i-1}, a)$, or $U_i = V_i \parallel W_i$ for some $V_i, W_i \in \pom^+$ and there are $r_i, s_i \in Q$ and $r_i', s_i' \in F$ such that $r_i \atrace{V_i} r_i'$, $s_i \atrace{W_i} s_i'$ and $q_i = \gamma(q_{i-1}, \mset{r_i,s_i})$.
We distinguish cases according to whether $n = 1$ or $n > 1$.
\begin{itemize}
    \item If $n = 1$, then we distinguish cases according to whether $U_1 = a$ or $U_1 = V_1 \parallel W_1$.
    In the first case, we have that $q' = \delta(q,a)$, so $U = U_1 = a \in \sem{e^{\emptyset}_{qq'}} \subseteq \sem{e^{Q''}_{qq'}}$.
    In the second case, we have that there exist $r,s'\in Q$ and $r', s' \in F$ such that $q' = \gamma(q, \mset{r,s})$, $r \atrace{V_1} r'$ and $s \atrace{W_1} s'$.
    Hence $V_1\in L_A(r) - \{ 1 \}$ and $W_1 \in L_A(s) - \{ 1 \}$.
    Furthermore, by fork-acyclicity we have that $r, s \prec_A q$, so there exist series-rational expressions $e_r^+, e_s^+ \in \terms$ such that $\sem{e_r^+} = L_A(r) - \{ 1 \}$ and $\sem{e_s^+} = L_A(s) - \{ 1 \}$.
    It follows that $V_1 \in \sem{e_r^+}$ and $W_1 \in \sem{e_s^+}$, and hence $V_1 \parallel W_1 \in \sem{e_r^+ \parallel e_s^+}$.
    We conclude that $U = U_1 = V_1 \parallel W_1 \in \sem{e^{\emptyset}_{qq'}} \subseteq \sem{e^{Q''}_{qq'}}$.

    \item If $n > 1$.
    Then, clearly, $Q''$ is non-empty; let $q''\in Q''$.
    Furthermore, let $I = \{ 0\leq i \leq n: q_i=q''\}$.
    If $I = \emptyset$, then by the induction hypothesis we have that $U\in\sem{e^{Q''-\{q''\}}_{qq'}}\subseteq\sem{e^{Q''}_{qq'}}$.

    Otherwise, suppose that $I \neq \emptyset$, say $I = \{i_1, \dots, i_k\}$ with $i_1 < \cdots < i_k$.
    Then, by the induction hypothesis, $U_{i_j+1}\cdots U_{i_{j+1}}\in\sem{e^{Q''-\{q''\}}_{q''q''}}$ for all $1\leq j < k$.
    Moreover, also by the induction hypothesis, $U_0\cdots U_{i_1}\in\sem{e^{Q''-\{q''\}}_{qq''}}$ and $U_{i_k+1}\cdots U_n\in\sem{e^{Q''-\{q''\}}_{q''q'}}$.
    Hence,
    \[U \in \sem{e^{Q''-\{q''\}}_{qq''} \cdot {(e^{Q''-\{q''\}}_{q''q''})}^{*} \cdot e^{Q''-\{q''\}}_{q''q'}} \subseteq\sem{e^{Q''}_{qq'}}\]
\end{itemize}

We prove the implication from right to left also with induction on the size of $Q''$.
Suppose that $U\in\sem{e^{Q''}_{qq'}}$.
If $Q''=\emptyset$, then, considering the structure of $e^{Q''}_{qq'}$, there are two cases:
\begin{itemize}
    \item If $U=a$ for some $a\in\Sigma$ such that $q'=\delta(q,a)$, then, clearly, $q\atrace{U}q'$.
    \item If there exist $r,s\in Q$ and $V,W\in\pom^+$ such that $\mset{r,s}\in\pi_A(q)$, $q'=\gamma(q,\mset{r,s})$, $U=V\parallel W$, $V\in\sem{e_r^+}$, and $W\in\sem{e_s^+}$, then (by the premise) there exist $r',s'\in F$ such that $r\atrace{V}r'$ and $s\atrace{W}s'$.
    Hence, $q \atrace{U} q'$.
\end{itemize}
Otherwise, suppose that $Q''$ is non-empty, and let $q''\in Q''$.
Then, again considering the structure of $e^{Q''}_{qq'}$, there are two cases:
\begin{itemize}
    \item If $U\in\sem{e^{Q''-\{q''\}}_{qq'}}$, then by the induction hypothesis $q\atrace{U}q'$.
    \item If $U\in\sem{e^{Q''-\{q''\}}_{qq''}\cdot {(e^{Q''-\{q''\}}_{q''q''})}^{*} \cdot e^{Q''-\{q''\}}_{q''q'}}$, then there exist $U'$, $U_1,\dots U_n$ ($n\geq 0$) and $U''$ such that $U=U'\cdot U_1\cdots U_n \cdot U''$, $U'\in\sem{e^{Q''-\{q''\}}_{qq''}}$, $U_i\in\sem{e^{Q''-\{q''\}}_{q''q''}}$ for all $1\leq i \leq n$, and $U''\in\sem{e^{Q''-\{q''\}}_{q''q'}}$.
 By the induction hypothesis, it follows that $q\atrace{U'}q''\atrace{U_1}\cdots\atrace{U_n}q''\atrace{U''}q'$, according to a series of paths that only visit states in $Q''-\{q''\}$, with the intermediate states of these paths all $q''$.
We thus conclude that $q\atrace{U}q'$ according to some path that only visits states in $Q''$. \qedhere
\end{itemize}
\end{arxivproof}

Using the auxiliary expressions $e^{Q''}_{qq'}$, we can now associate series-rational expressions $e_q, e_q^+\in\terms$ with every $q\in Q$, defining $e_q^+$ by $e_q^+=\sum_{q'\in\rho_A(q)\cap F}e^{\rho_A(q)}_{qq'}$ and $e_q=e_q^++[q\in F]$.
Note that $q\in\rho_A(q)$ and, by Lemma~\ref{lemma:reach-vs-hierarchy}, for all $q'\in Q$ such that $q'\prec_A q''$ for some $q''\in\rho_A(q)$ we have $q'\prec_A q$, and hence there exists, by induction, a series-rational expression $e_{q'}\in\terms$ such that $\sem{e_{q'}}=L_A(q')$.
So the expressions $e^{\rho_A(q)}_{qq'}$ are, indeed, defined in Definition~\ref{definition:autexpraux}.

\begin{corollary}
For every state $q\in Q$ we have $\sem{e_q^+}=L_A(q)-\{1\}$ and $\sem{e_q}=L_A(q)$.
\end{corollary}

\section{Discussion}%
\label{section:discussion}

Another automaton formalism for pomsets, \emph{branching automata}, was proposed by Lodaya and Weil~\cite{lodaya-weil-1998,lodaya-weil-2000}.
Branching automata define the states where parallelism can start (\emph{fork}) or end (\emph{join}) in two relations; pomset automata condense this information in a single function.
Lodaya and Weil also provided a translation of series-parallel expressions to branching automata, based on Thompson's construction~\cite{thompson-1968}, which relies on the fact that their automata encode transitions non-deterministically, i.e., as \emph{relations}.
Our Brzozowski-style~\cite{brzozowski-1964} translation, in contrast, directly constructs transition \emph{functions} from the expressions.
Lastly, their translation of branching automata to series-parallel expressions is only sound for a \emph{semantically} restricted class of automata, whereas our restriction is \emph{syntactic}.

Jipsen and Moshier~\cite{jipsen-moshier-2016} provided an alternative formulation of the automata proposed by Lodaya and Weil, also called \emph{branching automata}.
Their method to encode parallelism in these branching automata is conceptually dual to pomset automata: branching automata distinguish based on the target states of traces to determine the join state, whereas pomset automata distinguish based on the origin states of traces.
The translations of series-parallel expressions to branching automata and vice versa suffer from the same shortcomings as those by Lodaya and Weil, i.e., transition relations rather than functions and a semantic restriction on automata for the translation of automata to expressions.

Lodaya and Weil observed~\cite{lodaya-weil-2000} that the behaviour of their automata corresponds to $1$-safe Petri nets.
Since the behavior of their branching automata can be matched with our (bounded, fork-acyclic) pomset automata, we believe that $1$-safe Petri nets also correspond to our automata.
We opted to treat semantics of series-rational expressions in terms of automata instead of Petri nets to find more opportunities to extend to a coalgebraic treatment.
While the present paper does not reach this goal, we believe that our formulation in terms of states and transition functions offers some hope of getting there.

Prisacariu introduced \emph{Synchronous Kleene Algebra} (SKA)~\cite{prisacariu-2010}, extending Kleene Algebra with a \emph{synchronous composition} operator.
SKA differs from our model in that it assumes that all basic actions are performed in unit time, and that actors responsible for individual actions never idle.
In contrast, our (weak BKA-like) model makes no synchrony assumptions: expressions can be composed in parallel, and the relative timing of basic actions within those expressions is irrelevant for the semantics.
Prisacariu axiomatized SKA and extended it to \emph{Synchronous Kleene Algebra with Tests} (SKAT); others~\cite{broda-cavadas-ferreira-moreira-2015} proposed Brzozowski-style derivatives of SKA expressions and used them to test equivalence of SKA and SKAT terms.

\section{Further work}%
\label{section:further-work}

We plan to extend our results to semantics of series-parallel expressions in terms of downward-closed pomset languages, i.e., sets of pomsets that are closed under Gischer's subsumption order~\cite{gischer-1988}.
Such an extension would correspond to adding the weak exchange law (which relates sequential and parallel compositions), and thus yields an operational model for weak CKA\@.
We conjecture that no change to the automaton model is necessary to accommodate this generalization, just like Struth and Laurence suspect that the downward-closed semantics of series-parallel expressions can be captured by their non-downward closed semantics.

Our series-rational expressions do not include the parallel analogue of the Kleene star (sometimes called ``parallel star'', or ``replication'').
Future work could look into extending derivatives to include this operator, and relaxing fork-acyclicity to allow recovering expressions that include the parallel star from an automaton that satisfies this weaker restriction.

A classic result by Kozen~\cite{kozen-1994} axiomatizes language equivalence of rational expressions using Kleene's theorem~\cite{kleene-1956} and the uniqueness of minimal finite automata; consequently, the free model for KA can also be characterized in terms of rational languages.
It would be interesting to see if the same technique can be used (based on pomset automata) to show that the axioms of weak Bi-Kleene Algebra are a complete axiomatization of pomset language equivalence of series-rational expressions, and thus characterise the free weak Bi-Kleene Algebra (or even the free weak CKA) in terms of series-rational pomset languages.
Although an such a result was recently published~\cite{laurence-struth-2014}, it does not rely on an automaton model.

Brzozowski derivatives for classic rational expressions induce a coalgebra on rational expressions that corresponds to a finite automaton.
We aim to study series-rational expressions coalgebraically.
The first step would be to find the coalgebraic analogue of pomset automata such that language acceptance is characterized by the homomorphism into the final coalgebra.
Ideally, such a view of pomset automata would give rise to a decision procedure for equivalence of series-rational expressions based on coalgebraic bisimulation-up-to~\cite{rot-bonsangue-rutten-2013}.

Rational expressions can be extended with tests to reason about imperative programs equationally~\cite{kozen-1997}.
In the same vein, one can extend series-rational expressions with tests~\cite{jipsen-2014,jipsen-moshier-2016} to reason about parallel imperative programs equationally.
We are particularly interested in employing such an extension to extend the network specification language NetKAT~\cite{anderson-foster-guha-etal-2014} with primitives for concurrency so as to model and reason about concurrency within networks.

\bibliography{bibliography}

\ifarxiv%

\appendix

\section{Proofs of auxiliary lemmas}%
\label{appendix:proofs}

\congruencesound*
\begin{proof}
The proof of the first claim proceeds by induction on the construction of $\simeq$.
In the base, there are nine cases to consider.
\begin{itemize}
    \item If $e = f$, then the claim follows immediately.
    \item If $e = f + 0$, then $\sem{e} = \sem{f} \cup \sem{0} = \sem{f} \cup \emptyset = \sem{f}$.
    \item If $e = f + f$, then $\sem{e} = \sem{f} \cup \sem{f} = \sem{f}$.
    \item If $e = g_1 + g_2$ and $f = g_2 + g_1$, then $\sem{e} = \sem{g_1} \cup \sem{g_2} = \sem{g_2} \cup \sem{g_1} = \sem{f}$.
    \item If $e = g_1 + (g_2 + g_3)$ and $f = (g_1 + g_2) + g_3$, then $\sem{e} = \sem{g_1} \cup (\sem{g_2} \cup \sem{g_3}) = (\sem{g_1} \cup \sem{g_2}) \cup \sem{g_3} = \sem{f}$.
    \item If $e = 0 \cdot e_1$ and $f = 0$, then $\sem{e} = \sem{0} \cdot \sem{e_1} = \emptyset = \sem{f}$.
    \item If $e = e_1 \cdot 0$ and $f = 0$, then $\sem{e} = \sem{e_1} \cdot \sem{0} = \emptyset = \sem{f}$.
    \item If $e = 0 \parallel e_1$ and $f = 0$, then $\sem{e} = \sem{0} \parallel \sem{e_1} = \emptyset = \sem{f}$.
    \item If $e = e_1 \parallel 0$ and $f = 0$, then $\sem{e} = \sem{e_1} \parallel \sem{0} = \emptyset = \sem{f}$.
\end{itemize}

For the inductive step, there are six cases to consider.
\begin{itemize}
    \item If $e \simeq f$ because $e \simeq g$ and $g \simeq f$, then $\sem{e} = \sem{g}$ and $\sem{g} = \sem{f}$ by induction, thus $\sem{e} = \sem{f}$.
    \item If $e \simeq f$ because $f \simeq e$, then $\sem{f} = \sem{e}$ by induction, and thus $\sem{e} = \sem{f}$.
    \item If $e \simeq f$ because $e = e_1 + e_2$ and $f = f_1 + f_2$ with $e_1 \simeq f_1$ and $e_2 \simeq f_2$, then we know that $\sem{e_1} = \sem{f_1}$ and $\sem{e_2} = \sem{f_2}$ by induction.
    But then $\sem{e} = \sem{e_1} \cup \sem{e_2} = \sem{f_1} \cup \sem{f_2} = \sem{f}$.
    \item If $e \simeq f$ because $e = e_1 \cdot e_2$ and $f = f_1 \cdot f_2$ with $e_1 \simeq f_1$ and $e_2 \simeq f_2$, then we know that $\sem{e_1} = \sem{f_1}$ and $\sem{e_2} = \sem{f_2}$ by induction.
    But then $\sem{e} = \sem{e_1} \cdot \sem{e_2} = \sem{f_1} \cdot \sem{f_2} = \sem{f}$.
    \item If $e \simeq f$ because $e = e_1 \parallel e_2$ and $f = f_1 \parallel f_2$ with $e_1 \simeq f_1$ and $e_2 \simeq f_2$, then we know that $\sem{e_1} = \sem{f_1}$ and $\sem{e_2} = \sem{f_2}$ by induction.
    But then $\sem{e} = \sem{e_1} \parallel \sem{e_2} = \sem{f_1} \parallel \sem{f_2} = \sem{f}$.
    \item If $e \simeq f$ because $e = e_1^*$ and $f = f_1^*$ with $e_1 \simeq f_1$, then we know that $\sem{e_1} = \sem{f_1}$ by induction.
    But then $\sem{e} = \sem{e_1}^* = \sem{f_1}^* = \sem{f}$.
\end{itemize}

We now prove the second claim.
For the direction from left to right, note that $\sem{e} = \sem{0} = \emptyset$ by Lemma~\ref{lemma:congruence-sound}.
The direction from right to left proceeds by induction on $e$.
In the base, $e = 0$ and thus $e \simeq 0$ by reflexivity.
For the inductive step, there are three cases to consider.
\begin{itemize}
    \item If $e = e_1 + e_2$, then $\sem{e} = \sem{e_1} \cup \sem{e_2} = \emptyset$, therefore $\sem{e_1} = \emptyset$ and $\sem{e_2} = \emptyset$.
    But then, by induction, we have that $e_1 \simeq 0$ and $e_2 \simeq 0$.
    We then know that $e = e_1 + e_2 \simeq 0 + 0 \simeq 0$.
    \item If $e = e_1 \cdot e_2$, then $\sem{e} = \sem{e_1} \cdot \sem{e_2} = \emptyset$, therefore $\sem{e_1} = \emptyset$ or $\sem{e_2} = \emptyset$.
    But then, by induction, we have that $e_1 \simeq 0$ or $e_2 \simeq 0$.
    In either case, it follows that $e = e_1 \cdot e_2 \simeq 0$.
    \item If $e = e_1 \parallel e_2$, then $\sem{e} = \sem{e_1} \parallel \sem{e_2} = \emptyset$, therefore $\sem{e_1} = \emptyset$ or $\sem{e_2} = \emptyset$.
    But then, by induction, we have that $e_1 \simeq 0$ or $e_2 \simeq 0$.
    In either case, it follows that $e = e_1 \parallel e_2 \simeq 0$. \qedhere
\end{itemize}
\end{proof}

\reachvsreachable*
\begin{arxivproof}
We write $R_A(q)=\{ q' \in Q : \exists U \in \pom^+.\ q \atrace{U}_A q' \}$.

Suppose that $q' \in R_A(q) \cup \{ q, \bot \}$.
If $q' \in \{ q, \bot \}$, then $q' \in \rho_A(q) \cup \{ \bot \}$ by definition.
Otherwise, if $q' \in R_A(q)$, then $q \atrace{U} q'$ for some $U \in \pom^+$.
We show, more generally, that if $q'' \satrace{U} q'$ for some $q'' \in \rho_A(q) \cup \{ \bot \}$, then $q' \in \rho_A(q) \cup \{ \bot \}$, by induction on $U$.
In the base, $U = a \in \Sigma$ and thus $q'' = \delta(q', a)$.
But then $q' \in \rho_A(q) \cup \{ \bot \}$ immediately.
For the inductive step, there are two cases to consider.
\begin{itemize}
    \item If $U = V \cdot W$ with $V$ and $W$ smaller than $U$, then there exists an $r \in Q$ such that $q'' \atrace{V}_A r$ and $r \atrace{W}_A q'$.
    By induction, $r \in \rho_A(q) \cup \{ \bot \}$, and again by induction $q' \in \rho_A(q) \cup \{ \bot \}$.
    \item If $U = V \parallel W$ with $V$ and $W$ smaller than $U$, then there exist $r, s \in Q$ and $r', s' \in F$ such that $r \atrace{V}_A r'$ and $s \atrace{W}_A s'$, and $q' = \gamma(q'', \mset{r, s})$.
    Note that, in this case, $r', s' \neq \bot$.
    If $q' = \bot$, then $q' \in \rho_A(q) \cup \{ \bot \}$.
    Otherwise, $\mset{r, s} \in \pi_A(q'')$, and thus $q' \in \rho_A(q) \cup \{ \bot \}$.
\end{itemize}

For the second part of the claim, suppose that $\bot$ is the only state of $A$ with an empty language.
Let $q' \in \rho_A(q) \cup \{ \bot \}$.
If $q' = \bot$, then the claim follows.
For the case where $q' \in \rho_A(q)$, we prove that $q' \in R_A(q) \cup \{ q, \bot \}$, by induction on the construction of $\rho_A(q)$.
In the base, $q' = q$, and thus we find that $q' \in R_A(q) \cup \{ q, \bot \}$ immediately.

For the inductive step, there are two cases to consider.
For both cases, first note that if $q'' \in R_A(q')$ and $q' \in R_A(q)$, then $q'' \in R_A(q)$.
\begin{itemize}
    \item Suppose $q' = \delta(q'', a)$ for some $q'' \in \rho_A(q)$ and $a \in \Sigma$.
    By induction, we know that $q'' \in R_A(q) \cup \{ q, \bot \}$.
    If $q'' = \bot$, then $q' = \bot$ and the claim follows.
    Otherwise, note that $q' \in R_A(q'')$ by $q'' \atrace{a}_A q'$.
    If $q'' = q$, then $q' \in R_A(q)$ immediately.
    In the remaining case, $q'' \in R_A(q)$, and thus $q' \in R_A(q)$ by the observation above.
    \item Suppose $q' = \gamma(q'', \phi)$ for some $q'' \in \rho_A(q)$ and $\phi \in \pi_A(q'')$.
    By induction we know that $q'' \in R_A(q) \cup \{ q, \bot \}$.
    Since $\phi \in \pi_A(q'')$, we know that $q'' \neq \bot$ --- for otherwise $q' = \bot$ and thus $\phi \not\in \pi_A(q'')$.
    Furthermore, if $\phi = \mset{r, s}$ then $r, s \neq \bot$ and thus, since $\bot$ is the only state with an empty language, there exist $V, W \in \pom^+$ as well as $r', s' \in F$ such that $r \atrace{V}_A r'$ and $s \atrace{W}_A s'$; therefore, $q'' \atrace{U \parallel V}_A q'$ and thus $q' \in R_A(q'')$.
    If $q'' = q$, then $q' \in R_A(q)$ immediately.
    Otherwise, $q'' \in R_A(q)$, and thus $q' \in R_A(q)$ by the observation above. \qedhere
\end{itemize}
\end{arxivproof}

\congruentstateshalt*
\begin{proof}
The proof proceeds by induction on the construction of $\simeq$.
In the base, there are nine cases to consider.
\begin{itemize}
    \item Suppose that $e = f$, then the claim holds immediately.
    \item Suppose that $e = f + 0$.
    If $e \in \sacc$, then either $f \in \sacc$ or $0 \in \sacc$.
    Since the latter does not hold, $f \in \sacc$ follows.
    Also, if $f \in \sacc$, then $e = f + 0 \in \sacc$ immediately.
    \item Suppose that $e = f + f$.
    If $f + f \in \sacc$, then $f \in \sacc$ immediately; if $f \in \sacc$, then also $f + f \in \sacc$.
    \item Suppose that $e = g_1 + g_2$ and $f = g_2 + g_1$.
    If $g_1 + g_2 \in \sacc$, then $g_1 \in \sacc$ or $g_2 \in \sacc$; in either case, $g_2 + g_1 \in \sacc$.
    The claim in the other direction is proved similarly.
    \item Suppose that $e = g_1 + (g_2 + g_3)$ and $f = (g_1 + g_2) + g_3$.
    If $e \in \sacc$, then $g_1 \in \sacc$ or $g_2 + g_3 \in \sacc$, and thus one of $g_1, g_2, g_3$ must be in $\sacc$.
    But then $g_1 + g_2$ or $g_3$ must be in $\sacc$, and thus $f = (g_1 + g_2) + g_3 \in \sacc$.
    The proof in the other direction is similar.
    \item Suppose that $e = 0 \cdot e_1$ and $f = 0$.
    Then $e, f \not\in \sacc$ and thus the claim holds immediately.
    \item Suppose that $e = e_1 \cdot 0$ and $f = 0$.
    Then $e, f \not\in \sacc$ and thus the claim holds immediately.
    \item Suppose that $e = e_1 \parallel 0$ and $f = 0$.
    Then $e, f \not\in \sacc$ and thus the claim holds immediately.
    \item Suppose that $e = 0 \parallel e_1$ and $f = 0$.
    Then $e, f \not\in \sacc$ and thus the claim holds immediately.
\end{itemize}

For the inductive step, there are six cases to consider.
\begin{itemize}
    \item Suppose that $e \simeq f$ because $f \simeq e$; the claim (in both directions) then follows from the induction hypothesis by symmetry of mutual implication.
    \item Suppose that $e \simeq f$ because $e \simeq g$ and $g \simeq f$; the claim (in both directions) then follows from the induction hypothesis by transitivity of implication.
    \item Suppose that $e \simeq f$ because $e = e_1 + e_2$ and $f = f_1 + f_2$, with $e_1 \simeq f_1$ and $e_2 \simeq f_2$.
    If $e_1 + e_2 \in \sacc$, either $e_1 \in \sacc$ or $e_2 \in \sacc$.
    But then (by induction) either $f_1 \in \sacc$ or $f_2 \in \sacc$, thus $f_1 + f_2 \in \sacc$.
    The claim in the other direction is proved similarly.
    \item Suppose that $e \simeq f$ because $e = e_1 \cdot e_2$ and $f = f_1 \cdot f_2$, with $e_1 \simeq f_1$ and $e_2 \simeq f_2$.
    If $e \in \sacc$, then $e_1, e_2 \in \sacc$.
    But then (by induction) $f_1, f_2 \in \sacc$, thus $f_1 \cdot f_2 \in \sacc$.
    The claim in the other direction is proved similarly.
    \item Suppose that $e \simeq f$ because $e = e_1 \parallel e_2$ and $f = f_1 \parallel f_2$ with $e_1 \simeq f_1$ and $e_2 \simeq f_2$.
    If $e \in \sacc$, then $e_1, e_2 \in \sacc$.
    But then (by induction) $f_1, f_2 \in \sacc$, thus $f_1 \parallel f_2 \in \sacc$.
    The claim in the other direction is proved similarly.
    \item Suppose that $e \simeq f$ because $e = e_1^*$ and $f = f_1^*$ with $e_1 \simeq f_1$.
    Then $e, f \in \sacc$, and so the claim holds immediately in both directions. \qedhere
\end{itemize}
\end{proof}

\congruentstatesderive*
\begin{proof}
The proof of the first claim proceeds by induction on the construction of $\simeq$.
In the base, there are nine cases to consider.
\begin{itemize}
    \item If $e = f$, then $\ssderiv(e, a) = \ssderiv(f, a) \simeq \ssderiv(f, a)$.
    \item If $e = f + 0$, then $\ssderiv(e, a) = \ssderiv(f, a) + \ssderiv(0, a) = \ssderiv(f, a) + 0 \simeq \ssderiv(f, a)$.
    \item If $e = f + f$, then $\ssderiv(e, a) = \ssderiv(f, a) + \ssderiv(f, a) \simeq \ssderiv(f, a)$.
    \item If $e = g_1 + g_2$ and $f = g_2 + g_1$, then
    \[\ssderiv(e, a) = \ssderiv(g_1, a) + \ssderiv(g_2, a) \simeq \ssderiv(g_2, a) + \ssderiv(g_1, a) = \ssderiv(f, a)\]
    \item If $e = g_1 + (g_2 + g_3)$ and $f = (g_1 + g_2) + g_3$, then
    \[\ssderiv(e, a) = \ssderiv(g_1, a) + (\ssderiv(g_2, a) + \ssderiv(g_3, a)) \simeq (\ssderiv(g_1, a) + \ssderiv(g_2, a)) + \ssderiv(g_3, a) = \ssderiv(f, a)\]
    \item If $e = 0 \cdot e_1$ and $f = 0$, then $\ssderiv(e, a) = \ssderiv(0, a) \cdot e_1 + 0 = 0 \cdot e_1 + 0 \simeq 0 = \ssderiv(0, a)$.
    \item If $e = e_1 \cdot 0$ and $f = 0$, then $\ssderiv(e, a) = \ssderiv(e_1, a) \cdot 0 + e_1 \star \ssderiv(0, a) \simeq 0 = \ssderiv(0, a)$.
    \item If $e = e_1 \parallel 0$ and $f = 0$, then $\ssderiv(e, a) = e_1 \star \ssderiv(0, a) + 0 \simeq 0 = \ssderiv(0, a)$.
    \item If $e = 0 \parallel e_1$ and $f = 0$, then $\ssderiv(e, a) = 0 + e_2 \star \ssderiv(0, a) \simeq 0 = \ssderiv(0, a)$.
\end{itemize}
For the inductive step, there are six cases to consider.
\begin{itemize}
    \item If $e \simeq f$ because $f \simeq e$, then by induction we know that $\ssderiv(f, a) \simeq \ssderiv(e, a)$, thus $\ssderiv(e, a) \simeq \ssderiv(f, a)$.
    \item If $e \simeq f$ because $e \simeq g$ and $g \simeq f$, then by induction we know that $\ssderiv(e, a) \simeq \ssderiv(g, a)$ and $\ssderiv(g, a) \simeq \ssderiv(f, a)$ thus $\ssderiv(e, a) \simeq \ssderiv(f, a)$.
    \item If $e \simeq f$ because $e = e_1 + e_2$ and $f = f_1 + f_2$ such that $e_1 \simeq f_1$ and $e_2 \simeq f_2$, then by induction we know that $\ssderiv(e_1, a) \simeq \ssderiv(f_1, a)$ and $\ssderiv(e_2, a) \simeq \ssderiv(f_2, a)$, and thus
    \[\ssderiv(e, a) = \ssderiv(e_1, a) + \ssderiv(e_2, a) \simeq \ssderiv(f_1, a) + \ssderiv(f_2, a) = \ssderiv(f, a)\]
    \item If $e \simeq f$ because $e = e_1 \cdot e_2$ and $f = f_1 \cdot f_2$ such that $e_1 \simeq f_1$ and $e_2 \simeq f_2$, then by induction we know that $\ssderiv(e_1, a) \simeq \ssderiv(f_1, a)$ and $\ssderiv(e_2, a) \simeq \ssderiv(f_2, a)$, and thus
    \[\ssderiv(e, a) = \ssderiv(e_1, a) \cdot e_2 + e_1 \star \ssderiv(e_2, a) \simeq \ssderiv(f_1, a) \cdot f_2 + f_1 \star \ssderiv(f_2, a) = \ssderiv(f, a)\]
    where we also apply Lemma~\ref{lemma:congruent-states-halt}.
    \item If $e \simeq f$ because $e = e_1 \parallel e_2$ and $f = f_1 \parallel f_2$ such that $e_1 \simeq f_1$ and $e_2 \simeq f_2$, then by induction we know that $\ssderiv(e_1, a) \simeq \ssderiv(f_1, a)$ and $\ssderiv(e_2, a) \simeq \ssderiv(f_2, a)$, and thus
    \[\ssderiv(e, a) = e_2 \star \ssderiv(e_1, a) + e_1 \star \ssderiv(e_2, a) \simeq f_2 \star \ssderiv(f_1, a) + f_1 \star \ssderiv(f_2, a) = \ssderiv(f, a)\]
    where we also apply Lemma~\ref{lemma:congruent-states-halt}.
    \item If $e \simeq f$ because $e = e_1^*$ and $f = f_1^*$ such that $e_1 \simeq f_1$, then by induction we know that $\ssderiv(e_1, a) \simeq \ssderiv(f_1, a)$, and thus $\ssderiv(e, a) = \ssderiv(e_1, a) \cdot e_1^* \simeq \ssderiv(f_1, a) \cdot f_1^* = \ssderiv(f, a)$.
\end{itemize}

The proof of the second claim also proceeds by induction on the construction of $\simeq$.
In the base, there are nine cases to consider.
\begin{itemize}
    \item If $e = f$, then $\psderiv(e, \phi) = \psderiv(f, \phi) \simeq \psderiv(f, \phi)$.
    \item If $e = f + 0$, then $\psderiv(e, \phi) = \psderiv(f, \phi) + \psderiv(0, \phi) = \psderiv(f, \phi) + 0 \simeq \psderiv(f, \phi)$.
    \item If $e = f + f$, then $\psderiv(e, \phi) = \psderiv(f, \phi) + \psderiv(f, \phi) \simeq \psderiv(f, \phi)$.
    \item If $e = g_1 + g_2$ and $f = g_2 + g_1$, then
    \[\psderiv(e, \phi) = \psderiv(g_1, \phi) + \psderiv(g_2, \phi) \simeq \psderiv(g_2, \phi) + \psderiv(g_1, \phi) = \psderiv(f, \phi)\]
    \item If $e = g_1 + (g_2 + g_3)$ and $f = (g_1 + g_2) + g_3$, then
    \[\psderiv(e, \phi) = \psderiv(g_1, \phi) + (\psderiv(g_2, \phi) + \psderiv(g_3, \phi)) \simeq (\psderiv(g_1, \phi) + \psderiv(g_2, \phi)) + \psderiv(g_3, \phi) = \psderiv(f, \phi)\]
    \item If $e = 0 \cdot e_1$ and $f = 0$, then $\psderiv(e, \phi) = \psderiv(0, \phi) \cdot e_1 + 0 = 0 \cdot e_1 + 0 \simeq 0 = \psderiv(0, \phi)$.
    \item If $e = e_1 \cdot 0$ and $f = 0$, then $\psderiv(e, \phi) = \psderiv(e_1, \phi) \cdot 0 + e_1 \star \psderiv(0, \phi) \simeq 0 = \psderiv(0, \phi)$.
    \item If $e = 0 \parallel e_1$ and $f = 0$, then since $g, h \not\simeq 0$, in particular $\phi \not\simeq \mset{e_1, 0}$.
    We then find that $\psderiv(e, \phi) = 0 + e_1 \star \psderiv(0, \phi) + 0 \simeq 0 = \psderiv(0, \phi)$.
    \item If $e = e_1 \parallel 0$ and $f = 0$, then since $g, h \not\simeq 0$, in particular $\phi \not\simeq \mset{0, e_1}$.
    We then find that $\psderiv(e, \phi) = 0 + 0 + e_1 \star \psderiv(0, \phi) \simeq 0 = \psderiv(0, \phi)$.
\end{itemize}
For the inductive step, there are six cases to consider.
\begin{itemize}
    \item If $e \simeq f$ because $f \simeq e$, then by induction we know that $\psderiv(f, \phi) \simeq \psderiv(e, \phi)$, thus $\psderiv(e, \phi) \simeq \psderiv(f, \phi)$.
    \item If $e \simeq f$ because $e \simeq g$ and $g \simeq f$, then by induction we know that $\psderiv(e, \phi) \simeq \psderiv(g, \phi)$ and $\psderiv(g, \phi) \simeq \psderiv(f, \phi)$, thus $\psderiv(e, \phi) \simeq \psderiv(f, \phi)$.
    \item If $e \simeq f$ because $e = e_1 + e_2$ and $f = f_1 + f_2$ such that $e_1 \simeq f_1$ and $e_2 \simeq f_2$, then by induction we know that $\psderiv(e_1, \phi) \simeq \psderiv(f_1, \phi)$ and $\psderiv(e_2, \phi) \simeq \psderiv(f_2, \phi)$, and thus
    \[\psderiv(e, \phi) = \psderiv(e_1, \phi) + \psderiv(e_2, \phi) \simeq \psderiv(f_1, \phi) + \psderiv(f_2, \phi) = \psderiv(f, \phi)\]
    \item If $e \simeq f$ because $e = e_1 \cdot e_2$ and $f = f_1 \cdot f_2$ such that $e_1 \simeq f_1$ and $e_2 \simeq f_2$, then by induction we know that $\psderiv(e_1, \phi) \simeq \psderiv(f_1, \phi)$ and $\psderiv(e_2, \phi) \simeq \psderiv(f_2, \phi)$, and thus
    \[\psderiv(e, \phi) = \psderiv(e_1, \phi) \cdot e_2 + e_1 \star \psderiv(e_2, \phi) \simeq \psderiv(f_1, \phi) \cdot f_2 + f_1 \star \psderiv(f_2, \phi) = \psderiv(f, \phi)\]
    where we also apply Lemma~\ref{lemma:congruent-states-halt}.
    \item If $e \simeq f$ because $e = e_1 \parallel e_2$ and $f = f_1 \parallel f_2$ such that $e_1 \simeq f_1$ and $e_2 \simeq f_2$, then by induction we know that $\psderiv(e_1, \phi) \simeq \psderiv(f_1, \phi)$ and $\psderiv(e_2, \phi) \simeq \psderiv(f_2, \phi)$.
    Furthermore, $\phi \simeq \mset{e_1, e_2}$ if and only if $\phi \simeq \mset{f_1, f_2}$, and thus
    \begin{align*}
    \psderiv(e, \phi)
        &= [\mset{e_1, e_2} \simeq \phi] + e_2 \star \psderiv(e_1, \phi) + e_1 \star \psderiv(e_2, \phi)  \\
        &\simeq [\mset{f_1, f_2} \simeq \phi] + f_2 \star \psderiv(f_1, \phi) + f_1 \star \psderiv(f_2, \phi) \\
        &= \psderiv(f, \phi)
    \end{align*}
    where we also apply Lemma~\ref{lemma:congruent-states-halt}.
    \item If $e \simeq f$ because $e = e_1^*$ and $f = f_1^*$ such that $e_1 \simeq f_1$, then by induction we know that $\psderiv(e_1, \phi) \simeq \psderiv(f_1, \phi)$, and thus $\psderiv(e, \phi) = \psderiv(e_1, \phi) \cdot e_1^* \simeq \psderiv(f_1, \phi) \cdot f_1^* = \psderiv(f, \phi)$.
\end{itemize}

The last claim is shown by induction on $e$.
In the base, $e = 0$, $e = 1$ or $e = a$, and thus $\psderiv(e, \phi) = 0 = \psderiv(e, \psi)$.
For the inductive step, there are four cases to consider.
\begin{itemize}
    \item If $e = e_1 + e_2$, then $\psderiv(e, \phi) = \psderiv(e_1, \phi) + \psderiv(e_2, \phi) = \psderiv(e_1, \psi) + \psderiv(e_2, \psi) = \psderiv(e, \psi)$.
    \item If $e = e_1 \cdot e_2$, then $\psderiv(e, \phi) = \psderiv(e_1, \phi) \cdot e_2 + e_1 \star \psderiv(e_2, \phi) = \psderiv(e_1, \psi) \cdot e_2 + e_1 \star \psderiv(e_2, \psi) = \psderiv(e, \psi)$.
    \item If $e = e_1 \parallel e_2$, then $\psderiv(e, \phi) = [\mset{e_1, e_2} \simeq \phi] + e_2 \star \psderiv(e_1, \phi) + e_1 \star \psderiv(e_2, \phi) = [\mset{e_1, e_2} \simeq \psi] + e_2 \star \psderiv(e_1, \psi) + e_1 \star \psderiv(e_2, \psi)$.
    \item If $e = e_1^*$, then $\psderiv(e, \phi) = \psderiv(e_1, \phi) \cdot e_1^* = \psderiv(e_1, \psi) \cdot e_1^* = \psderiv(e, \psi)$. \qedhere
\end{itemize}
\end{proof}

\syntracevsquotsyntrace*
\begin{proof}
The proof for the first claim proceeds by induction on $U$.
In the base, $U = a \in \Sigma$.
If $e \satrace{U} f$, then $f = \ssderiv(e, a)$.
But then, since $\qssderiv([e], a) = [f]$ by definition of $\qssderiv$, we find that $[e] \qsatrace{U} [f]$.

For the inductive step, there are two cases to consider.
\begin{itemize}
    \item If $U = V \cdot W$, with $V$ and $W$ smaller than $U$, then there exists a $g \in \terms$ such that $e \satrace{V} g$ and $g \satrace{W} f$.
    By induction we find that $[e] \qsatrace{V} [g]$ and $[g] \qsatrace{W} [f]$, and thus $[e] \qsatrace{U} [f]$.
    \item If $U = V \parallel W$, with $V$ and $W$ smaller than $U$, then there exist $g, h \in \terms$ and $g', h' \in \sacc$ such that $g \satrace{V} g'$ and $h \satrace{W} h'$, and $f = \psderiv(e, \mset{g, h})$.
    Note that, by Lemma~\ref{lemma:trace-deconstruct-base}, $g, h \not\simeq 0$.
    By induction, we find $[g] \qsatrace{V} [g']$ and $[h] \qsatrace{W} [h']$, with $[g'], [h'] \in \qsacc$.
    Since $\qpsderiv([e], \mset{[g], [h]}) = [\psderiv(e, \mset{g, h})] = [f]$, we find that $[e] \qsatrace{U} [f]$.
\end{itemize}

The proof of the second claim also proceeds by induction on $U$.
In the base, $U = a \in \Sigma$.
If $[e] \satrace{U} [f]$, then $[f] = \qssderiv([e], a)$.
We choose $f' = \ssderiv(e, a)$ to find that $e \satrace{U} f$ and $f \simeq f'$ by Lemma~\ref{lemma:congruent-states-derive}.

For the inductive step, there are two cases to consider.
\begin{itemize}
    \item If $U = V \cdot W$, with $V$ and $W$ smaller than $U$, then there exists a $g \in \terms$ such that $[e] \qsatrace{V} [g]$ and $[g] \qsatrace{W} [f]$.
    By induction, we find $g', f'' \in \terms$ such that $e \satrace{V} g'$ and $g \simeq g'$ as well as $g \satrace{V} f''$ and $f \simeq f''$.
    By Lemma~\ref{lemma:congruent-states-progress}, we find $f' \in \terms$ such that $g' \satrace{V} f'$ and $f'' \simeq f$.
    In total, we have that $e \satrace{U} f'$ with $f \simeq f'$.
    \item If $U = V \parallel W$, with $V$ and $W$ smaller than $U$, then there exist $[g], [h] \in \qterms$ and $[g'], [h'] \in \qsacc$ such that $[g] \qsatrace{V} [g']$ and $[h] \qsatrace{W} [h']$, and $[f] = \qpsderiv([e], \mset{[g], [h]})$.
    By induction, we find that $h'', g'' \in \terms$ such that $g \satrace{V} g''$ and $h \satrace{W} h''$, with $g' \simeq g''$ and $h' \simeq h''$, with $g'', h'' \in \sacc$ by Lemma~\ref{lemma:congruent-states-halt}.
    But then $g, h \not\simeq 0$ by Lemma~\ref{lemma:trace-deconstruct-base}.
    We can choose $f' = \psderiv(e, \mset{g, h})$ to find that that $e \satrace{U} f'$ and $f \simeq f'$ by Lemma~\ref{lemma:congruent-states-derive}. \qedhere
\end{itemize}
\end{proof}

\depthvscongruence*
\begin{arxivproof}
First, if $e \simeq 0 \simeq e'$, then $\depth{e} = 0 = \depth{e'}$ by definition.
The proof for the remaining cases proceeds by induction on the construction of $\simeq$.
In the base, there are nine cases to consider.
\begin{itemize}
    \item If $e = e'$, then the claim holds immediately.
    \item If $e = e' + 0$, then $\depth{e} = \max(\depth{e'}, \depth{0}) = \depth{e'}$.
    \item If $e = e' + e'$, then $\depth{e} = \max(\depth{e'}, \depth{e'}) = \depth{e'}$.
    \item If $e = e_1 + e_2$ and $e' = e_2 + e_1$, then $\depth{e} = \max(\depth{e_1}, \depth{e_2}) = \max(\depth{e_2}, \depth{e_1}) = \depth{e'}$.
    \item If $e = e_1 + (e_2 + e_3)$ and $e' = (e_1 + e_2) + e_3$, then $\depth{e} = \max(\depth{e_1}, \depth{e_2}, \depth{e_3}) = \depth{e'}$.
\end{itemize}
For the inductive step, there are six cases to consider.
\begin{itemize}
    \item If $e \simeq e'$ because $e' \simeq e$, then $\depth{e'} = \depth{e}$ by induction and so the claim follows.
    \item If $e \simeq e'$ because $e \simeq e'' \simeq e'$, then $\depth{e} = \depth{e''} = \depth{e'}$ by induction and so $\depth{e} = \depth{e'}$.
    \item If $e \simeq e'$ because $e = e_1 + e_2$ and $e' = e_1' + e_2'$ with $e_1 \simeq e_1'$ and $e_2 \simeq e_2'$, then $\depth{e_1} = \depth{e_1'}$ and $\depth{e_2} = \depth{e_2'}$ by induction, thus $\depth{e} = \max(\depth{e_1}, \depth{e_2}) = \max(\depth{e_1'}, \depth{e_2'}) = \depth{e'}$.
    \item If $e \simeq e'$ because $e = e_1 \cdot e_2$ and $e' = e_1' \cdot e_2'$ with $e_1 \simeq e_1'$ and $e_2 \simeq e_2'$, then $\depth{e_1} = \depth{e_1'}$ and $\depth{e_2} = \depth{e_2'}$ by induction, thus $\depth{e} = \max(\depth{e_1}, \depth{e_2}) = \max(\depth{e_1'}, \depth{e_2'}) = \depth{e'}$.
    \item If $e \simeq e'$ because $e = e_1 \parallel e_2$ and $e' = e_1' \parallel e_2'$ with $e_1 \simeq e_1'$ and $e_2 \simeq e_2'$, then $\depth{e_1} = \depth{e_1'}$ and $\depth{e_2} = \depth{e_2'}$ by induction, thus $\depth{e} = \max(\depth{e_1}, \depth{e_2}) + 1 = \max(\depth{e_1'}, \depth{e_2'}) + 1 = \depth{f}$.
    \item If $e \simeq e'$ because $e = e_1^*$ and $e' = e_1'^*$, then by induction $\depth{e} = \depth{e_1} = \depth{e_1'} = \depth{e'}$. \qedhere
\end{itemize}
\end{arxivproof}

\fi%

\end{document}
